\documentclass{amsart}

\RequirePackage{amsthm,amsmath,amsfonts,amssymb,comment}
\RequirePackage[numbers,sort&compress]{natbib}
\RequirePackage[colorlinks,citecolor=blue,urlcolor=blue]{hyperref}
\RequirePackage{graphicx}

\usepackage{amsthm}
\usepackage{hyperref}
\hypersetup
{colorlinks = true,
linkcolor = red, 
anchorcolor = red, 
citecolor = blue, 
filecolor = blue,
urlcolor = blue}
\usepackage{hypcap}
\usepackage{cleveref}
\usepackage{lipsum}
\usepackage{enumerate}
\usepackage{enumitem}
\usepackage{amsmath,amsfonts,amssymb,mathrsfs, amscd,amsthm,amsbsy,amsxtra,bbm,bm, epsf,calc,comment, xcolor}
\usepackage[toc,page]{appendix}
\usepackage{color}
\usepackage{datetime}
\usepackage{latexsym}
\usepackage[english]{babel}
\usepackage{graphicx}
\usepackage{epsfig}
\usepackage{dsfont}
\usepackage{tikz}

\usepackage{soul}
\usepackage{multirow}

\usepackage{algorithmic} 
\usepackage{algorithm} 
\usepackage{dsfont}

\theoremstyle{plain}

\numberwithin{equation}{section}
\newtheorem{theorem}{Theorem}[section]
\newtheorem{lemma}[theorem]{Lemma}
\newtheorem{definition}[theorem]{Definition}
\newtheorem{proposition}[theorem]{Proposition}
\newtheorem{corollary}[theorem]{Corollary}
\newtheorem{example}[theorem]{Example}

\newtheorem{assumption}[theorem]{Assumption}

\theoremstyle{remark}
\newtheorem{remark}[theorem]{Remark}

\definecolor{mygreen}{rgb}{0.1,0.75,0.2}

\DeclareMathOperator*{\argmin}{argmin}
\DeclareMathOperator*{\esssup}{ess\, sup}

\DeclareSymbolFont{bbold}{U}{bbold}{m}{n}
\DeclareSymbolFontAlphabet{\mathbbold}{bbold}

\newcommand{\spt}{\textup{spt}}

\begin{document}

\title[Extension of coupling via OT]{Extension of coupling via the Projection of Optimal Transport}

\thanks{*: Equal contribution}
\author{Jakwang Kim*}
\address{School of Data Science, The Chinese University of Hong Kong, Shenzhen, Guangdong, 518172, P.R. China.}
\email{jakwangkim@cuhk.edu.cn}
\author{Young-Heon Kim}
\address{Department of Mathematics, University of British Columbia, 1984 Mathematics Road, Vancouver, British Columbia, V6T 1Z2, Canada.}
\email{yhkim@math.ubc.ca}
\author{Chan Park*}
\address{Department of Statistics, University of Illinois Urbana-Champaign, 136 Computing Applications Building,
605 E Springfield Ave, Champaign, IL 61820, USA.}
\email{parkchan@illinois.edu}

\date{\today}
\thanks{
JK is supported by CUHK-SZ start-up UDF03004229.
YHK is partially supported by the Natural Sciences and Engineering Research Council of Canada (NSERC), with Discovery Grant RGPIN-2019-03926 and RGPIN-2025-06747, as well as Exploration Grant (NFRFE-2019-00944) from the New Frontiers in Research Fund (NFRF). YHK is also a member of the Kantorovich Initiative (KI), which is supported by the PIMS Research Network (PRN) program of the Pacific Institute for the Mathematical Sciences (PIMS). We thank PIMS for their generous support. This research is partially done while YHK was visiting Korea Advanced Institute of Science and Technology (KAIST), and he thanks their generous support and excellent research environment. 
\copyright 2026 by the authors. All rights reserved.}

\begin{abstract}
In many statistical settings, two types of data are available: coupled data, which preserve the joint structure among variables but are limited in size due to cost or privacy constraints, and marginal data, which are available at larger scales but lack joint structure. Since standard methods require coupled data, marginal information is often discarded. We propose a fully nonparametric procedure that integrates decoupled marginal data with a limited amount of coupled data to improve the downstream analysis. The approach can be understood as an extension of coupling via projection in optimal transport. Specifically, the estimator is a solution for the optimal transport projection over the space of probability measures, which genuinely provides a natural geometric interpretation. Not only is its stability established, but its sample complexity is also derived using recent advances in statistical optimal transport. In addition to this, we present its explicit formula based on ``shadow," a notion introduced by Eckstein and Nutz. Furthermore, the estimator can be approximated in almost linear time and in parallel by entropic shadow, which demonstrates the theoretical and practical strengths of our methods. Lastly, we present experiments with real and synthetic data to justify the performance of our method.
\end{abstract}


\maketitle
\tableofcontents

\section{Introduction} \label{sec: intro}

Many statistical methods aim to elucidate relationships among a set of variables $\bm{Z}$. Standard approaches---regression, machine learning, and related tools---assume access to independent and identically distributed (i.i.d.) observations $\bm{Z}_1, \ldots, \bm{Z}_m$, each a complete realization of the random vector $\bm{Z}$ drawn from its joint distribution. This intact joint structure is crucial: it ensures that the observed sample reflects the true relationships in $\bm{Z}$, up to sampling error, and theoretical guarantees of statistical methods are expressed in terms of the sample size, to the extent that certain properties (e.g., consistency, central limit theorems) are attained as $m \to \infty$. Throughout this manuscript, we refer to data preserving this joint structure as \textit{coupled data}.

Unfortunately, fully coupled data are often available only for a limited fraction of the population, typically due to privacy, administrative constraints, or high collection costs. For instance, the U.S. Census Bureau’s American Community Survey (ACS) 2018 5-year Public Use Microdata Sample (PUMS) contains approximately $m \approx 6.13 \times 10^5$ individual records for Illinois. Each record includes multiple attributes---such as demographic and socioeconomic characteristics---jointly observed, allowing estimation of relationships among variables (e.g., conditional probabilities). Although this sample is relatively large, it represents only about 4.8\% of Illinois’s total population ($n \approx 1.28 \times 10^7$), leaving the vast majority unrepresented in the joint model. In contrast, marginal or low-dimensional aggregate statistics are often available at much larger scales. The ACS 2018 5-year Summary File, for example, provides population-level marginal distributions of each characteristic in Illinois, covering the full population. However, these data do not preserve the joint structure present in PUMS, leaving the dependencies between variables unobserved. We refer to such datasets—containing only marginal information without the coupling structure—as \textit{marginal data}.

In practice, marginal data are often regarded as ``auxiliary'' because they cannot be readily utilized by off-the-shelf statistical methods, which typically require coupled data as input. Consequently, researchers often restrict their analysis to the coupled dataset, justifying this choice with the assumption that the data are a representative sample of the population. While this justification is well-grounded, relying solely on coupled data inherently discards other available information. This observation leads to the overarching question of this paper: ``Can we leverage the information in marginal data to enhance the estimation of relationships among variables?''

The challenge of enhancing statistical inference by integrating detailed but small-scale coupled data with aggregate but large-scale marginal data is well-established across statistics and econometrics. In the survey sampling literature, this is primarily framed as \textit{calibration} \citep{Deville1992, Wu2001model, Wu2003, Montanari2005, Kim2010}, a method that re-weights survey samples so that weighted aggregates match known population totals. In the missing data literature, marginal data provide crucial information for imputing unobserved components of the coupled data, thereby enabling recovery of the full joint distribution \citep{Sadinle2019sequentially, Akande2021leveraging}. Similarly, in econometrics \citep{Imbens1994, Hellerstein1999, Nevo2003, Chen2008}, marginal data (often referred to as ``macro'' data) serve as moment restrictions within likelihood or generalized method of moments (GMM; \citealp{GMM1982}) frameworks to improve the precision of micro-level parameter estimates. More recently, researchers have utilized deep generative methods to create synthetic populations that mirror coupled data while respecting the structure of marginal data \citep{Qian2025, Sane2025}.

Crucially, although these approaches incorporate marginal data as constraints or as known ground truth, they differ from our objective in several important respects. First, calibration methods primarily aim to improve estimation of population means for selected variables, rather than to recover complex relationships among variables; moreover, some calibration techniques can produce negative or unstable weights, leading to counterintuitive or unreliable results. Second, in the missing data literature, marginal data are leveraged under specific parametric assumptions about the missingness mechanism and often require computationally intensive Bayesian inference to reconstruct the joint distribution, including careful justification of prior selection. Third, GMM approaches with marginal constraints, as commonly used in econometrics, generally target estimation of specific low-dimensional parameters (e.g., regression coefficients), rather than recovery of the full joint distribution. Finally, although recent deep generative methods share our goal of reconstructing lost coupling structure, their black-box nature often limits interpretability and obscures the underlying statistical principles. 

\subsection{Overview of Our Approach} Our approach complements the existing literature by providing a statistically principled framework that explicitly reconstructs the joint coupling structure without relying on black-box architectures or restrictive parametric assumptions.

We formally introduce the problem and our main contributions as follows. For a space $\mathcal{S}$, $\mathcal{P}(\mathcal{S})$ denotes the set of probability measures on $\mathcal{S}$. Given a probability $\pi^0 \in \mathcal{P}(\mathcal{X}_1 \times \dots \times \mathcal{X}_K)$, let $\left\{ \bm{Z}_j=(Z_{1j}, \dots, Z_{Kj}): j=1, \dots, m \right\}$ be the set of coupled data i.i.d. from $\pi^0$, of which marginal distributions are $\mu^i \in \mathcal{P}(\mathcal{X}_i)$ for $i=1,\dots,K$. Denoting by $\bm{\mu} := (\mu^1, \dots, \mu^K)$, the vector of $K$-marginals, it can be rewritten shortly as
\begin{equation}\label{eq: hypothesis}
    \pi^0 \in \Pi(\bm{\mu}):= \Pi(\mu^1, \dots, \mu^K),
\end{equation}
which is the set of joint distributions over $\mathcal{X}_1 \times \dots \times \mathcal{X}_K$ whose marginals are $\bm{\mu}$. 
On the other hand, there are marginal data $\{ X_{1j} : j=1, \dots, n \}, \dots, \{ X_{Kj} : j=1, \dots, n \}$. We assume that they are originally i.i.d. from the same $\pi^0$ but observed only marginally, hence $X_{ij} \sim \mu^i$ independently for $i=1,\dots, K$.

In the ACS application of Section \ref{sec:data}, each $\bm{Z}_j$ represents the $j$th individual's jointly observed record of $K=5$ socioeconomic characteristics: health insurance status ($Z_{1j}$), working-age indicator ($Z_{2j}$), sex ($Z_{3j}$), race ($Z_{4j}$), and education level ($Z_{5j}$), drawn from the PUMS dataset of size $m \approx 6.13 \times 10^5$. In addition, $\{ X_{ij} : j=1, \dots, n \}$ represents the population-level distribution of the $i$th characteristic. For example, $\{ X_{3j} \}$ is the marginal distribution of sex, obtained from the Summary File covering $n \approx 1.28 \times 10^7$ individuals in Illinois. Lastly, $\pi^0$ represents the unknown joint distribution of these characteristics within the Illinois population.

In practice, even when $n$ is much larger than $m$ (as in the ACS example above) and the marginal data contain information about $\pi^0$, they are often discarded because many existing statistical approaches require the input data to exhibit coupling structures. The goal is to reconstruct the \emph{coupling of marginal datasets by using a relatively small amount of coupled data}, so that the marginal data can be incorporated into the analysis and overall inference can be improved.

To formalize, define empirical measures based on the samples
\[
    \pi^0_m:= \frac{1}{m}\sum_{j=1}^m \delta_{\bm{Z}_j}, \quad \mu^i_n:= \frac{1}{n}\sum_{j=1}^n \delta_{X_{ij}} \text{ for $i=1, \dots, K$}, \quad \bm{\mu}_n=(\mu^1_n, \dots, \mu^K_n).
\]
Our main problem is then expressed as:
\begin{align}\label{eq:goal}
    \textit{optimally reconstruct $\pi^0$ from $\pi^0_m$ and $\bm{\mu}_n$.}
\end{align}
Although $\pi^0 \in \Pi(\bm{\mu})$, $\pi^0_m \not\in \Pi(\bm{\mu}_n)$ almost surely. However, it is still reasonable to expect that under the hypothesis \eqref{eq: hypothesis} there exists a $\hat{\pi} \in\Pi(\bm{\mu}_n)$ that is similar to $\pi^0_m$, thereby providing a good candidate for a coupling among $\mu^i$'s.

Our approach to \eqref{eq:goal} is to find a point in $\Pi(\bm{\mu}_n)$ which is closest to $\pi^0_m$ in the Wasserstein distance ($\mathcal{W}_p$; see \Cref{subsec: notation}); it is indeed a metric projection. Such an \emph{optimal transport projection} has been considered in another context by  \citet[Remark 4.2]{Eckstein_Nutz_2022} (and \cite{bayraktar_StabilitySampleComplexity_2025}) for a different purpose: they considered a particular type of such a projection called {\em shadow} and used it in studying the stability of Sinkhorn iterations. Our problem, instead, is statistical, and we use such an optimal transport projection as a method to extend the given sampled coupling to a coupling of the entire set of marginals, serving as an estimator of the hidden ground truth coupling; we view it as a key novelty of our present paper to frame marginal data integration as an OT projection problem. We use shadow method of \cite{Eckstein_Nutz_2022} to characterize our statistical estimator. In addition, we establish new sample-complexity bounds tailored to the two-sample structure $(m,n)$ of our problem, derive the limit distribution of the estimator under a finite-support assumption, and provide associated confidence sets based on the recent works \cite{Sommerfeld_Munk_inference_OT2017, limit_random_lp22,liu2023asymptotic}.

The projection is formulated  (see \citet[Remark 4.2]{Eckstein_Nutz_2022}) as 
\begin{equation}\label{eq: Wasserstein projection}
    \pi^0_m \longmapsto \hat{\pi} := \argmin_{\pi' \in \Pi(\bm{\mu}_n)} \mathcal{W}_p(\pi', \pi^0_m).
\end{equation}
In the ACS application, $\hat{\pi}$ is the estimated joint distribution of the five socioeconomic variables for the full Illinois population. It is constructed by finding the coupling in $\Pi(\bm{\mu}_n)$---the set of joint distributions whose marginals match those observed in the Summary File---that is closest to the empirical joint distribution $\pi^0_m$ from the smaller PUMS dataset in terms of the Wasserstein metric. In other words, $\hat{\pi}$ extends the coupling structure learned from the PUMS to the entire population while respecting the population-level marginal constraints provided by the Summary File.

Since $\mathcal{W}_p$ is continuous and $\Pi(\bm{\mu}_n)$ is weakly compact, \eqref{eq: Wasserstein projection} always attains a solution. Under the hypothesis \eqref{eq: hypothesis}, $\pi^0$ is the unique minimizer of the Wasserstein projection of $\pi^0$ onto $\Pi(\bm{\mu})$, therefore it is expected that $\hat{\pi}$ is close to $\pi^0$ as $\bm{\mu}_n$ is close to $\bm{\mu}$. (However, note that $\hat \pi$ is not unique in typical cases as the marginals $\mu^i_n$'s are discrete measures; see \Cref{rem:nonunique}.)

Main contributions of this work can be summarized as follows. 
\begin{enumerate}
    \item[(i)] A fully nonparametric approach for \eqref{eq:goal} by using \eqref{eq: Wasserstein projection}. In words, we frame the reconstruction problem as a minimization problem over Wasserstein space. Leveraging the geometry of this space, our framework provides a principled alternative to existing methods with several advantages. Unlike traditional survey calibration, it goes beyond simple re-weighting to reconstruct the full joint distribution geometrically, avoiding negative weights and the limitation of focusing only on population means. Compared with missing data methods, it does not require parametric assumptions about the missingness mechanism or the specification of subjective priors. Finally, unlike deep generative approaches, the Wasserstein-based projection replaces black-box optimization with a interpretable and mathematically grounded geometric problem.
    \item[(ii)] We establish several theoretically and practically significant properties of the estimator. In particular, we prove consistency and derive upper and lower bounds of its convergence rate, characterizing its finite-sample complexity (\Cref{theorem: Quantitative stability} and \Cref{cor: sample complexity} in \Cref{section: main results}). In the special case where the support of a probability measure is finite, we further describe the estimator's limiting distribution and corresponding confidence sets, enabling formal statistical inference (\Cref{thm: limit distribution of projection} in \Cref{section: finite support case}).
    \item[(iii)] We also provide concrete and efficient computation methods for $\hat{\pi}$ using \emph{shadow} and \emph{entropic regularization}, demonstrating that it is efficiently computable. These methods are illustrated in \Cref{sec:sim} and \Cref{sec:data} through simulation studies and real-world applications.
\end{enumerate}


\subsection{Organization}
We will briefly discuss the background of optimal transport and prerequisite notation in \Cref{sec:preplim}. In \Cref{section: main results}, we will introduce the optimal transport projection framework and its theoretical results in detail. We continue to discuss the limit distribution and the confidence sets of the proposed estimator for finite support case in \Cref{section: finite support case}. In \Cref{sec:sim} and \Cref{sec:data}, synthetic and real data analysis will support the practicality of our method. Lastly, we will conclude this paper in \Cref{sec: conclusion}. More details on the proofs and simulation results will be discussed in \Cref{sec:appendix:deferred proofs} and \Cref{sec:appendix:additional sim}, respectively.

\section{Preliminaries}\label{sec:preplim}

\subsection{Review: Optimal Transport}

In the present paper, we propose an optimal transport (OT) method to reconstruct couplings from partial data. OT was proposed by \citet{monge1781memoire} and developed by \citet{kantorovich1942translocation, kantorovich1948problem}, who developed the Kantorovich relaxation, for which he received the Nobel Prize in economic sciences. Since then, OT has interwoven partial differential equations, geometry, and probability~\cite{Brenier91, mccann1994convexity, Caffarelli_contraction, Benamou_brener2000, Otto2000GeneralizationOA, otto2001geometry, Lott2004RicciCF, Figalli2010Invent, Kim2010JEMS, Figall2011duke, Caffarelli2013Reine, Figalli2015CMP, GOZLAN20173327}. Recently, OT plays a central role in data sciences: optimization of mean field two-layer neural networks~\cite{mei2018mean, chizat2018global, sirignano2020mean, Rotskoff2022CPAM}, diffusion models~\cite{de2021diffusion, wang2021deep, hamdouche2023generative}, variational inference~\cite{lambert2022variational, diao2023forward}, and adversarial training~\cite{pydi2021the, Trillos2020AdversarialCN}.

Late 18th the Monge problem \cite{monge1781memoire} was introduced to find a transport map to minimize the cost with the pushforward condition. Later, Kantorovich studied its convex relaxation, in which the solution space is extended to the set of couplings \cite{kantorovich1942translocation, kantorovich1948problem}.

Here, we focus on a special case of the Kantorovich problem, which is obtained when a Polish space $\mathcal{X}=\mathcal{Y}$ and $c=d^p$ where $d$ is a metric on $\mathcal{X}$. The \emph{$p$-Wasserstein distance} for $1 \leq p \leq \infty$, or \emph{Kantorovich-Rubinstein distance}, is defined over $\mathcal{P}_p(\mathcal{X})$, the set of probability measures with finite $p$-th order moment, as follows:
\[
\begin{aligned}
    {W}_p(\mu, \nu)& := \inf_{\gamma \in \Pi(\mu, \nu)} \left( \int_{\mathcal{X} \times \mathcal{X}} d(x,y)^p d\gamma(x,y) \right)^{\frac{1}{p}} \text{ for $1 \leq p < \infty$},\\
    W_\infty(\mu, \nu) &:= \inf_{\gamma \in \Pi(\mu, \nu)} \esssup_{(X,Y) \sim \gamma} d(X,Y).
\end{aligned}
\]
It is well known that $\mathcal{P}_p(\mathcal{X})$ equipped with $W_p$ is a Polish space: in particular, $(\mathcal{P}_2(\mathcal{X}), {W}_2)$ exhibits a Riemannian structure \cite{otto2001geometry}.

An important ingredient used in this work is the recent progress of statistical optimal transport, which, roughly speaking, aims at understanding the behaviour of $W_p(\mu_n,\mu)$ where $\mu_n$ is an empirical distribution of i.i.d. samples drawn from $\mu$. Since \citet{dudley69} paved this field, there have been numerous papers in this direction \cite{MR2280433, MR2861675, Dereich_Scheutzow_Schottstedt2013, MR3189084, NF_AG_rate_Wasserstein, MR4359822, merigot2021non, MR4255123, MR4441130}.

An interesting extension of OT is \emph{multimarginal optimal transport} (MOT). Given $\mu^i \in \mathcal{P}(\mathcal{X}_i)$ for $i=1, \dots, K$ and $c : \mathcal{X}_1 \times \dots \times \mathcal{X}_K \to \mathbb{R} $, MOT is defined as
\[
    \min_{\pi \in \Pi(\mu^1,\dots, \mu^K)} \int c(x_1, \dots, x_K) d\pi(x_1, \dots ,x_K).
\]
Since \citet{Gangbo1998OptimalMF} studied the multimarginal Monge problem, there have been a lot of subsequent papers establishing its geometric and analytic aspects \cite{Kim2013MultimarginalOT, KimPass-SIAM2014, MR3423275, MR3482267, MR4809472, pass2025dynamicalformulationmultimarginaloptimal}. MOT problems have also been found in other fields as well: density functional theory in physics \cite{seidl2007strictly, buttazzo2012optimal, MR3020313, mendl2013kantorovich, MR3314838}, theoretical economics \cite{MR2110613, MR2564439, Carlier2010MatchingFT}, Wasserstein barycenters \cite{MR2801182, cuturi2014fast, benamou2015iterative, MR3423268, srivastava2018scalable, delon2020wasserstein} and adversarial training in machine learning \cite{jakwang2023JMLR, jakwang_2024existence, jakwang_2024optimal}.

\subsection{Notation}\label{subsec: notation}
Let $\mathcal{Z} := \mathcal{X}_1 \times \dots \times \mathcal{X}_K$, and $\bm{x}:=(x_1, \dots, x_K) \in \mathcal{Z}$. Recall $\bm{\mu} := (\mu^1, \dots, \mu^K)$ be the vector of $K$-marginals, and similarly $\bm{\mu}_n := (\mu^1_n, \dots, \mu^K_n)$ where $\mu^i_n$'s are empirical measures of $\mu^i$'s. For simplicity, we use $(\otimes \bm{\mu}):=\mu^1 \otimes \dots \otimes \mu^K$, $\Pi(\bm{\mu}):= \Pi(\mu^1, \dots, \mu^K)$ and $\Pi(\bm{\mu}, \pi^0):=\Pi(\mu^1, \dots, \mu^K, \pi^0)$; here, $\Pi(\mu^1, \dots, \mu^K)$ denotes the set of probability measures with marginals $\mu^1, \dots, \mu^K$.

Given metric spaces $(\mathcal{X}_i, d_i)$ for $i=1, \dots, K$, define the \emph{separable} product metric $d_{\mathcal{Z},p}$ over $\mathcal{Z}$ as 
\begin{equation}\label{eq: metric over Z}
    d_{\mathcal{Z},p}(\bm{x}, \bm{z}):=
    \begin{cases}
        \left( \sum_{i=1}^K d_i(x_i, z_i)^p\right)^\frac{1}{p} &\text{ for $p \in[1,\infty)$},\\
        \max_{i=1, \dots, K} d_i(x_i, z_i)  &\text{ for $p =\infty$}.
    \end{cases} 
\end{equation}
Unless otherwise noted, the $p$-Wasserstein distances on $\mathcal{X}_i$ and $\mathcal{Z}$ are defined with respect to $d_i$ and $d_{\mathcal{Z},p}$, respectively. For notational convenience and clarity, we use $\mathcal{W}_p(\pi, \pi')$ to the $p$-Wasserstein distance on $\mathcal{Z}$ only, and
\[
    W_p(\bm{\mu}, \bm{\nu}):= \left( \sum_{i=1}^K W^p_p(\mu^i, \nu^i) \right)^{\frac{1}{p}} \text{for $p \in [1,\infty)$}, \quad W_\infty(\bm{\mu}, \bm{\nu}):= \max_{i=1, \dots, K} W_\infty(\mu^i, \nu^i).
\]

On the space $\mathcal{Z} \times \mathcal{Z}$, $\mathrm{P}_{x_i}$ and $\mathrm{P}_{z_i}$ denote  projections onto $\mathcal{X}_i$ of the first $\mathcal{Z}$, and that of the second, respectively. Similarly, we use $\mathrm{P}_{\bm{x}}$ and $\mathrm{P}_{\bm{z}}$ to denote projections onto the first $\mathcal{Z}$ and the second, respectively.

\section{Main results
}\label{section: main results}
\subsection{Optimal transport projection}

Since $\pi^0 \in \Pi(\bm{\mu})$, note that $\pi^0$ is the unique minimizer for the Wasserstein projection onto $\Pi(\bm{\mu})$, that is,
\[
    \pi^0 =\argmin_{\pi' \in \Pi(\bm{\mu})} \mathcal{W}_p(\pi', \pi^0). 
\]
It is expected that a solution $\hat{\pi}$ for \eqref{eq: Wasserstein projection}, even though it is nonunique in typical cases (see \cref{rem:nonunique}), will converge to $\pi^0$ as $m$ and $n$, the numbers of samples, grow. Not only is this true, but also the upper and lower bounds of the $p$-Wasserstein distance between $\hat{\pi}$ and $\pi^0$ can be obtained as follows as a simple consequence of \cite[Lemma 3.2]{Eckstein_Nutz_2022}. Bounds are described in terms of the distances between the true and empirical distributions only.

\begin{theorem}[Stability]\label{theorem: Quantitative stability}
Let $\nu^i_m$ be the $i$-th marginal of $\pi^0_m$ for $i=1, \dots, K$, and $\bm{\nu}_m:=(\nu^1_m, \dots, \nu^K_m)$. Then
\begin{equation*}
    W_p(\bm{\mu}, \bm{\mu}_n) \leq \mathcal{W}_p(\hat{\pi}, \pi^0) \leq W_p(\bm{\mu}_n, \bm{\nu}_m)  +  \mathcal{W}_p(\pi^0_m, \pi^0).
\end{equation*}

\begin{proof}
Since $\hat\pi \in \Pi(\bm{\mu}_n)$, a lower bound follows as
\begin{align*}
    \mathcal{W}_p(\hat{\pi}, \pi^0) \geq \min_{\pi \in \Pi(\bm{\mu}_n)} \mathcal{W}_p(\pi, \pi^0) =  W_p(\bm{\mu}, \bm{\mu}_n),
\end{align*}
where the last equation is by \cite[Lemma 3.2]{Eckstein_Nutz_2022} and \eqref{eq: hypothesis}.

An upper bound is attained by \cite[Lemma 3.2]{Eckstein_Nutz_2022} as
\begin{align*}
    \mathcal{W}_p(\hat{\pi}, \pi^0) \leq  \mathcal{W}_p(\hat{\pi}, \pi^0_m) + \mathcal{W}_p(\pi^0_m, \pi^0)&= \min_{\pi \in \Pi(\bm{\mu}_n)} \mathcal{W}_p(\pi, \pi^0_m) + \mathcal{W}_p(\pi^0_m, \pi^0)\\
    &=W_p(\bm{\mu}_n, \bm{\nu}_m) +  \mathcal{W}_p(\pi^0_m, \pi^0).
\end{align*}
\end{proof}
\end{theorem}

\begin{remark}[Statistical interpretation of the stability bounds]
The two terms on the right-hand side of the upper bound in Theorem \ref{theorem: Quantitative stability} have distinct statistical interpretations. The term $W_p(\bm{\mu}_n, \bm{\nu}_m)$ measures how well the marginals of the small coupled data ($\bm{\nu}_m$) match the marginals of the large decoupled data ($\bm{\mu}_n$); this term is small when the coupled sample is representative of the population marginals. The term $\mathcal{W}_p(\pi^0_m, \pi^0)$ measures the quality of the coupled empirical measure as an approximation of the true joint distribution $\pi^0$; this is the standard sampling error from the coupled data alone, which vanishes as $m$ grows. The lower bound $W_p(\bm{\mu}, \bm{\mu}_n)$ shows that the estimator cannot be better than the marginal convergence rate determined by the large dataset of size $n$. Together, these bounds confirm qualitatively that the proposed estimator benefits from both data sources.
\end{remark}


Combined with \cite[Theorem 1]{NF_AG_rate_Wasserstein}, \Cref{theorem: Quantitative stability} yields the upper and lower bounds of the sample complexity of $\hat{\pi}$. Here we assume that all dimensions are sufficiently larger than $p$ for simplicity.

\begin{corollary}[Sample complexity] \label{cor: sample complexity}
Assume that $\mathcal{X}_i= \mathbb{R}^{d_i}$ and $d_i >2p$ for $1 \leq i \leq K$, also that there is some $q > \frac{p \sum d_i }{\sum d_i - p}$ such that the $q$-th moment of $\pi^0$ (hence all $\mu^i$'s) is finite. Then there are constants $C_1, C_2, C_3$ which do not depend on $n,m$ but on $p, \sum d_i, q$ and $q$-th moments of each of the measures such that
\[
    C_1\left( \sum_{i=1}^K n^{-\frac{p}{d_i}} \right)^\frac{1}{p} \leq \mathbb{E} \mathcal{W}_p(\hat{\pi}, \pi^0)  \leq  C_2 m^{-\frac{1}{d_1 + \dots + d_K}} + C_3 \left( \sum_{i=1}^K n^{-\frac{p}{d_i}} \right)^\frac{1}{p}.
\]    
\end{corollary}

{\begin{remark}[Curse of dimensionality]
The corollary provides the quantitative estimation error by the two sample sizes $m$ and $n$. Folklore in the OT community says the rate of convergence of empirical distributions in the Wasserstein distance suffers the curse of dimensionality: $O( n^{-\frac{1}{d}} )$ where $d$ is the ambient dimension. Since the marginal terms $n^{-p/d_i}$ are negligible when $n \gg m$, the overall error is dominated by the coupled data term $m^{-1/(d_1 + \cdots + d_K)}$. This reconfirms that the estimation error exhibits the curse of dimensionality, which conforms with the fact that estimating the joint distribution becomes increasingly difficult in higher dimensions.

On the other hand, if the support of the measure enjoys intrinsic low dimensionality, the exponents can be improved. \citet{JW_FB_sample_rates} proved the optimal rate of convergence of empirical distributions in the Wasserstein distance is $O(n^{-\frac{1}{k}})$ where $k$ relates to the \emph{Wasserstein dimension} provided that its support is compact. Thus, the sample complexity of $\hat{\pi}$ can be improved if $\pi^0$ possesses the intrinsic low dimensionality.
\end{remark}
}

The next natural question is how to compute $\hat{\pi}$. Rewriting \eqref{eq: Wasserstein projection}, it is equivalent to the following MOT problem:
\begin{equation}\label{eq: MOT formulation of Wasserstein projection}
\begin{aligned}
    \min_{\gamma \in \Pi(\bm{\mu}_n, \pi^0_m)} & \left\{ \int  d_{\mathcal{Z},p}(\boldsymbol{x}, \boldsymbol{z})^p d\gamma \right\}^{\frac{1}{p}} \text{ for $p \in [1,\infty)$},\\
    \min_{\gamma \in \Pi(\bm{\mu}_n, \pi^0_m)} & \max_{i=1, \dots, K} \esssup_{(X_i, Z_i) \sim \gamma^i} d_i(X_i, Z_i) \text{ for $p=\infty$}
\end{aligned}
\end{equation}
where the marginal constraints should be understood as $(\mathrm{P}_{\bm{x}})_\#(\gamma) = \pi'$ for some $\pi' \in \Pi(\bm{\mu}_n)$ and $(\mathrm{P}_{\bm{z}})_\#(\gamma) = \pi^0_m$. Note that $\gamma^i= (\mathrm{P}_{x_i}, \mathrm{P}_{z_i})_\# \gamma \in \Pi(\mu^i_n, \nu^i_m)$. Denoting an optimal multimarginal coupling of \eqref{eq: MOT formulation of Wasserstein projection} by $\hat{\gamma}$, an  optimal solution $\hat{\pi}$ of \eqref{eq: Wasserstein projection} is obtained from $\hat{\gamma}$ by $\hat{\pi} := (\mathrm{P}_{\boldsymbol{x}})_\# (\hat{\gamma})$.

\begin{remark}\label{rem:nonunique}
Note for instance the above explains nonuniqueness of $\hat \pi$ as an optimizer $\hat \gamma$ of the MOT problem \eqref{eq: MOT formulation of Wasserstein projection} is generally not unique.
\end{remark}

OT (linear programming in general) is solvable in $O(n^3)$ computational complexity by the simplex method or the Hungarian method. In particular, \eqref{eq: MOT formulation of Wasserstein projection} requires $O(n^{3(K+1)})$ complexity. Note that solving the generic MOT is NP-hard as the number of marginals increases \cite{barycenter_NPhard2022}.

Using the separability of $d_{\mathcal{Z}, p}$, however, $\hat{\pi}$ can be computed efficiently based on \emph{shadow} as proposed by \citet{Eckstein_Nutz_2022}:

\begin{theorem}[Shadow \cite{Eckstein_Nutz_2022}]\label{thm: shadow}
Let $\hat{\gamma}_i$ be an optimal coupling for $W_p(\mu^i_n, \nu^i_m)$, which can be written as $\hat{\gamma}_i(dx_i, dz_i) = \nu_m^i(dz_i) \hat{\kappa}^i(dx_i |z_i)$ by disintegration formula. Then
\begin{equation*}
    \hat{\gamma}(d \boldsymbol{x}, d\boldsymbol{z}) = \pi^0_m(d\boldsymbol{z}) \kappa(d\boldsymbol{x} | \boldsymbol{z})
\end{equation*}
is optimal for \eqref{eq: MOT formulation of Wasserstein projection} where $\hat{\kappa}(d\boldsymbol{x} | \boldsymbol{z}):=  \kappa_1(dx_1 |z_1) \otimes \dots \otimes \kappa_K(dx_K |z_K)$. In particular, an optimal $\hat{\pi}$ (the shadow of $\pi^0_m$ onto $\Pi(\bm{\mu}_n)$) for \eqref{eq: Wasserstein projection} is obtained by
\[
    \hat{\pi}(d\boldsymbol{x}) = \int_{\mathcal{Z}} \pi^0_m(d\boldsymbol{z}) \hat{\kappa}(d\boldsymbol{x} | \boldsymbol{z}).
\] 
Furthermore, $\eqref{eq: Wasserstein projection}=\eqref{eq: MOT formulation of Wasserstein projection}=
    W_p(\bm{\mu}_n, \bm{\nu}_m)$.
\end{theorem}

\begin{proof}
See \cite[Lemma 3.2]{Eckstein_Nutz_2022}.
\end{proof}

\begin{remark}\label{rem:nonunique-projection}
Note that optimal solutions to the optimal transport projection problem \eqref{eq: Wasserstein projection}, equivalently \eqref{eq: MOT formulation of Wasserstein projection}, are not necessarily given by a shadow; 
see \citet[Remark 4.2]{Eckstein_Nutz_2022}.
\end{remark}


The computational benefit of shadow is straightforward. Each $\hat{\gamma}_i$ can be computed with $O(n^3)$ computational cost. Overall, the total computational cost of $\hat{\gamma}$ is $O(Kn^3)$, which is much less than the $O(n^{3(K+1)})$ complexity of \eqref{eq: MOT formulation of Wasserstein projection}. In addition to this efficiency, parallel computation for $\hat{\gamma}$ is applicable in this case, which improves the speed of the computation.

\subsection{Entropic regularization}
There is an estimator with almost linear time complexity to approximate OT, the so-called \emph{entropic regularization} method, proposed by \citet{cuturi2013sinkhorn}. This framework allows one to compute the approximate solution in almost linear time \cite{altschuler2017near, lin2022complexity}.

Fix $\eta > 0$, called the entropic parameter. 
For our MOT problem \eqref{eq: MOT formulation of Wasserstein projection}, the usual entropic regularized version is defined as follows:
\begin{align*}
    \min_{\gamma \in \Pi(\bm{\mu}_n, \pi^0_m)} &\int  d_{\mathcal{Z},p}(\boldsymbol{x}, \boldsymbol{z})^p d\gamma  + \eta \mathrm{KL} \left( \gamma \, \Vert\,  (\otimes \bm{\mu}_n) \otimes \pi^0_m \right) \text{ for $p \in [1,\infty)$},\\ 
    \min_{\gamma \in \Pi(\bm{\mu}_n, \pi^0_m)} &\max_{i=1, \dots, K} \esssup_{(X_i, Z_i) \sim \gamma} d_i(X_i, Z_i) + \eta \mathrm{KL} \left( \gamma \, \Vert\,  (\otimes \bm{\mu}_n) \otimes \pi^0_m \right) \text{ for $p=\infty$}.
\end{align*}
Due to the strict convexity of the KL divergence, the entropic MOT provides a unique optimal solution. Also, importantly, the entropic MOT boosts the speed of computation. However, this approach does not take advantage of the separability of the cost $d_{\mathcal{Z},p}(\boldsymbol{x}, \boldsymbol{z})^p$.

\begin{remark}
The generic $K$ marginals MOT requires $O( \eta^{-2}K^3 n^{K} \log n)$ computational cost if each of marginals has $O(n)$ points \cite[Theorem 16]{lin2022complexity}.    
\end{remark}

Let us introduce \emph{entropic shadow}, which allows us to leverage the separability structure significantly.
Consider the following entropic regularized version of \eqref{eq: MOT formulation of Wasserstein projection}: 
\begin{equation}\label{eq: entropic MOT formulation of Wasserstein projection}
\begin{aligned}
    \min_{\gamma \in \Pi(\bm{\mu}_n, \pi^0_m)} & \sum_{i=1}^K \int d_i(x_i, z_i)^p d\gamma^i + \eta \mathrm{KL}(\gamma^i \, \Vert\, \mu^i_n \otimes \nu^i_m) \text{ for $p \in [1,\infty)$},\\
    \min_{\gamma \in \Pi(\bm{\mu}_n, \pi^0_m)} & \max_{i=1, \dots, K} \esssup_{(X_i, Z_i) \sim \gamma^i} d_i(X_i, Z_i) + \eta \mathrm{KL}(\gamma^i \, \Vert\, \mu^i_n \otimes \nu^i_m) \text{ for $p=\infty$}.
\end{aligned}
\end{equation}
Denoting by $\hat{\gamma}^{i,\eta} \in \Pi(\mu^i_n, \nu^i_m)$ the (unique) optimal entropic coupling for
\[
    \int d_i(x_i, z_i)^p d\gamma^i + \eta \mathrm{KL}(\gamma^i \, \Vert\, \mu^i_n \otimes \nu^i_m) \text{ or } \esssup_{(X_i, Z_i) \sim \gamma^i} d_i(X_i, Z_i) + \eta \mathrm{KL}(\gamma^i \, \Vert\, \mu^i_n \otimes \nu^i_m),
\]
define the kernel $\hat{\kappa}^{i,\eta}(dx_i |z_i)$ such that $\hat{\kappa}^{i,\eta}(dx_i |z_i)\nu^i(dz_i) = \hat{\gamma}^{i, \eta}(dx_i, d z_i)$. Similarly to \Cref{thm: shadow}, the optimal entropic MOT coupling for \eqref{eq: entropic MOT formulation of Wasserstein projection}, denoted by $\hat{\gamma}^\eta$, is obtained as
\[
    \hat{\gamma}^\eta (d\boldsymbol{x}, d \boldsymbol{z} ) = \pi^0_m(d \boldsymbol{z}) \hat{\kappa}^\eta(d \boldsymbol{x} |\boldsymbol{z})
\]
where $\hat{\kappa}^\eta = \hat{\kappa}^{1,\eta}(dx_1 |z_1) \otimes \dots \otimes \hat{\kappa}^{K,\eta}(dx_K |z_K)$. Then the entropic estimator $\hat{\pi}^\eta$ is achieved by
\begin{align}\label{eqn:entropic-shadow}
    \hat{\pi}^\eta(d\boldsymbol{x}) = \int_{\mathcal{Z}} \pi^0_m(d\boldsymbol{z}) \hat{\kappa}^\eta(d\boldsymbol{x} | \boldsymbol{z}).
\end{align} 
The name entropic shadow is rooted from $\pi^0_m$ which is obtained via the entropic regularized $p$-Wasserstein distance.

\begin{remark}[Nonuniqueness]
Notice that since the entropic marginal couplings $\hat \gamma^{i, \eta}$'s are uniquely determined, the entropic shadow is also uniquely determined; however, this does not imply uniqueness of an optimizer for  \eqref{eq: entropic MOT formulation of Wasserstein projection} as one consider a counterexample similar to that given in \cite[Remark 4.2]{Eckstein_Nutz_2022}.    
\end{remark}

A strong advantage of entropic shadow is its computational speed and efficiency. Note that it can be computed in parallel as before, and each $\hat{\gamma}^i$ can be computed by Sinkhorn algorithm, of which computational cost is almost linear. Overall, it decreases the computational cost to $\widetilde{O} (K n^{2})$, which is much less than $\widetilde{O} (K^3 n^K)$.

A drawback of entropic approach is that that $\hat \pi^\eta$ does not recover the true $\pi^0$ in general as $n,m \to \infty$ as long as $\eta > 0$, due to the entropic bias. One could ask whether the bias vanishes as the regularization disappears. Such a convergence is well known in the literature. Under mild conditions, \emph{$\Gamma$-convergence} guarantees the convergence of entropic estimator to $\pi^0$; see \Cref{def: Gamma convergence}.

\begin{remark}\label{rem:Gamma-convergence}
Let us discuss a brief history for the $\Gamma$-convergence of the entropic OT to the vanilla OT. \citet{Christian2012} proves the $\Gamma$-convergence by relating it to the large deviation theory under the assumption that $c$ is nonnegative and lower semicontinuous. \citet{CarlierDuvalPeyreSchmitzer2017} prove the same result for $c=d^2$ with the emphasis of numerical applications. \citet{BerntonGhosalNutz2022} prove the same result with the large deviation theory and provide characterization of the rate function in terms of Kantorovich potentials based on geometric argument. \citet{NutzWiesel2022} prove the convergence of entropic dual potentials to Kantorovich potentials.
\end{remark}

\begin{proposition}[$\Gamma$-convergence]\label{prop: Gamma convergence}
As $\eta \to 0$, \eqref{eq: entropic MOT formulation of Wasserstein projection} $\Gamma$-converges to \eqref{eq: MOT formulation of Wasserstein projection}. Furthermore, there is a subsequence of solutions of \eqref{eq: entropic MOT formulation of Wasserstein projection} which converges to a solution for \eqref{eq: MOT formulation of Wasserstein projection}.
\end{proposition}

\begin{proof}
See \cite[Section 5]{Nutz_introEOT2022}.  
\end{proof}

\begin{corollary}[Vanishing bias]\label{cor: consistency of entropic estimator}
There is a subsequence of $\hat{\pi}^\eta$ converging to $\pi^0$ weakly as $\eta \to 0$, and $n,m\to \infty$.    
\end{corollary}

\begin{proof}
By \cite[Theorem 1.4]{GHOSAL2022109622}, $\hat{\gamma}^\eta \to \gamma^\eta$ weakly as $n,m \to \infty$, hence so as $\hat{\pi}^\eta \to \pi^\eta$ weakly. At the same time, \Cref{prop: Gamma convergence} implies that there is a subsequence of $\pi^\eta$ converging to $\pi^0$. Choosing $(n_k,m_k,\eta_k)$ accordingly, therefore, one can find a subsequence of $\hat{\pi}^{\eta_k}_{n_k,m_k}$ converging to $\pi^0$ weakly.
\end{proof}

\section{Limit distribution and confidence set: finite support case}\label{section: finite support case}

Discrete data are common in many applications. For example, in the ACS dataset, all five variables (health insurance status, working-age indicator, sex, education level, and race) are categorical. Binary variables take two values, while race has four categories, meaning that the joint distribution spans at most $2 \times 2 \times 2 \times 4 \times 2 = 64$ cells; thus, $\pi^0$ is inherently a finitely supported distribution. More generally, survey and administrative data in social science, public health, and government often involve categorical or discretized responses, making finite support a natural assumption.

Motivated by this, we focus on the case where  $\pi^0$ (hence each of the $\mu^i$'s) has finite support, aligning with this practical motivation. The goal of this section is to derive the limit distribution of an estimator $\hat{\pi}$. For the finite support case, our problem reduces to a finite-dimensional linear program; hence, understanding the limit distribution of $\hat{\pi}$ requires studying the asymptotic behavior of the corresponding linear programming, which has been recently investigated by \citet{Sommerfeld_Munk_inference_OT2017}, \citet{limit_random_lp22} and \citet{liu2023asymptotic}.

Let $\mathcal{X}_i=\spt(\mu^i) = \{ x_{i1}, \dots, x_{i{s_i}} \}$ and $s_i := |\mathcal{X}_i|$ for $i=1, \dots, K$ and $s_0 := |\spt(\pi^0)| \leq s_1\cdots s_K$. Also, we will use $S_+:= s_0 + s_1 + \dots + s_K$ and $S_*:=s_0s_1\dots s_K$. After vectorizing all the objects, \eqref{eq: MOT formulation of Wasserstein projection} becomes a finite-dimensional linear program:
\begin{equation}\label{def: finite linear primal}
    \min_{\gamma \in \mathbb{R}^{S_*}} \langle \boldsymbol{c}, \gamma \rangle,\quad \textrm{s.t.}\ \bm{A}\gamma= \bm{b}_{n,m},\ \gamma\geq\bm{0}
\end{equation}
where $\boldsymbol{c} \in \mathbb{R}^{S_*}$ is the vectorization of the $s_1 \times \dots \times s_K \times s_0$ cost tensor (defined as $d_{\mathcal{Z}, p}^p$), $\bm{A} \in  \mathbb{R}^{S_+ \times S_*}$ is the projection matrix corresponding to the marginal constraints, and 
\begin{equation*}
    \bm{b}_{n,m}:=
\begin{bmatrix}
\bm{\mu}_n\\
\pi^0_m
\end{bmatrix},
\quad 
\bm{b} := 
\begin{bmatrix}
\bm{\mu}\\
\pi^0
\end{bmatrix}\in \mathbb{R}^{s_1} \times \dots \times \mathbb{R}^{s_K} \times \mathbb{R}^{s_0}
\end{equation*}
are the vector of marginal distributions.

\begin{remark}
For generic MOT, a canonical way to vectorize the tensors $\bm{c}$ and $\gamma$ is to write them in lexicographical order.

\end{remark}

\begin{example}[Marginal constraint vectors in the ACS application]
To illustrate the structure of $\bm{b}_{n,m}$ and $\bm{b}$, consider a simplified version of the ACS application with $K=2$ variables: sex ($X_3$, binary: $s_1=2$) and race ($X_4$, 4-level: $s_2=4$). The joint support has $s_0 \leq s_1 \cdot s_2 = 8$ cells. Then
\[
\bm{b}_{n,m} =
\begin{bmatrix}
\mu^{3}_n \\[3pt]
\mu^{4}_n \\[3pt]
\pi^0_m
\end{bmatrix}
=
\begin{bmatrix}
(\hat{p}_{\text{Male}},\, \hat{p}_{\text{Female}}) \\[3pt]
(\hat{p}_{\text{White}},\, \hat{p}_{\text{Black}},\, \hat{p}_{\text{Asian}},\, \hat{p}_{\text{Others}}) \\[3pt]
(\hat{p}_{\text{Male,White}},\, \hat{p}_{\text{Male,Black}},\, \ldots,\, \hat{p}_{\text{Female,Others}})
\end{bmatrix}
\in \mathbb{R}^{2} \times \mathbb{R}^{4} \times \mathbb{R}^{8},
\]
where $\hat{p}$ denotes empirical proportions: the first two blocks are estimated from the Summary File (i.e., marginal data), and the last block from the PUMS (i.e., coupled data). The population counterpart $\bm{b}$ has the same structure with true proportions replacing $\hat{p}$. 
\end{example}

The terms $\bm{c}$ and $\bm{A}$ are fixed and only $\bm{b}_{n,m}$ varies, so we denote the set of solutions for \eqref{def: finite linear primal} by $\gamma(\bm{b}_{n,m})$. \citet{limit_random_lp22} and \citet{liu2023asymptotic} study the limit distribution of $\gamma(\bm{b}_{n,m})$ for the case
\begin{equation}\label{eq: limit of b_n}
    r_{n,m} (\bm{b}_{n,m} - \bm{b}) \overset{d}{\longrightarrow} \mathbb{G}
\end{equation}   
for some random process $\mathbb{G}$ with rate $r_{n,m} \to \infty$ as $n,m \to \infty$. For the limit distribution of $\gamma(\bm{b}_{n,m})$, the following assumptions on \eqref{def: finite linear primal} should be imposed.

\begin{assumption}{\cite[Assumption 1]{liu2023asymptotic}}
\label{assumption: finite support case}
The constraint matrix $\bm{A}$ has full rank, the optimal solution set $\gamma(\bm{b})$ is nonempty and bounded, and the Slater's condition \cite[Section 5.2.3]{boyd2004convex} is satisfied, i.e., there exists $\gamma_0 \in \mathbb{R}^{S_*}$ such that $\gamma_0>0,\, \bm{A} \gamma_0=\bm{b}$.
\end{assumption}


\begin{remark}[Slater's condition]
In the statistical context, the Slater's condition means that there must exist a joint distribution over $\mathcal{Z}$ with strictly positive probability on every cell of the product support.
 In the ACS application, this requires every combination of health insurance status, working age, sex, race, and education to have positive probability in the Illinois population, which is a mild and verifiable condition given the large population size.
\end{remark}

Let $\hat{\pi}(\bm{b}_{n,m})$ denote a solution for the projection, i.e., $\hat{\pi}(\bm{b}_{n,m}) = (\mathrm{P}_{\bm{x}})_\# \gamma(\bm{b}_{n,m})$. Also, let $\mathbb{G}_{\bm{\mu}, \pi^0}$ be the limit distribution of $\bm{b}_{n,m}$ given in \eqref{eq: limit of b_n}: see \eqref{eq: limit of mu and pi^0} for more details. 
We apply 
\cite[Theorem 6.1]{limit_random_lp22} and \cite[Theorem 4]{liu2023asymptotic}
 to get the limit distribution of $\hat{\pi}(\bm{b}_{n,m})$ as follows.

\begin{theorem}\label{thm: limit distribution of projection}
Assume that $\frac{m}{n+m} \to \lambda \in (0, 1)$ as $n,m \to \infty$, $\spt(\pi^0)$ is finite and $c=d_{\mathcal{Z},p}^p$ in \eqref{eq: metric over Z}. Let $\hat{\pi}(\bm{b}_{n,m})$ be a solution for \eqref{eq: Wasserstein projection} with input $\bm{\mu}_n$ and $\pi^0_m$. Then
\[
    \sqrt{\frac{nm}{n+m}} \left( \hat{\pi}(\bm{b}_{n,m}) -  \pi^0 \right) \overset{d}{\longrightarrow} (\mathrm{P}_{\boldsymbol{x}})_\# (p^*(\mathbb{G}_{\bm{\mu}, \pi^0}))
\]
where $p^*(\mathbb{G}_{\bm{\mu}, \pi^0})$ is the set of optimal solutions for the following linear program:
\begin{equation*}
    \min_{p \in \mathbb{R}^{S_*}} \langle \boldsymbol{c}, p \rangle,\quad \textrm{s.t.}\ \bm{A}p= \mathbb{G}_{\bm{\mu}, \pi^0},\ p_{ij} \geq 0 \text{ for all } (ij) \notin \spt((\mathrm{Id}, \mathrm{Id})_\#(\pi^0)).
\end{equation*}

\begin{proof}
\Cref{lemma: assumption check} shows that \eqref{def: finite linear primal} satisfies \Cref{assumption: finite support case}. Then, by \Cref{lemma: aymptotic of linear programming}, we have
\begin{equation}\label{eqn:limit-coupling}
    \sqrt{\frac{nm}{n+m}} ( \gamma(\bm{b}_{n,m}) - \gamma(\bm{b}) ) \overset{d}{\longrightarrow} p^*(\mathbb{G}_{\bm{\mu}, \pi^0})
\end{equation}
where $\mathbb{G}_{\bm{\mu}, \pi^0}$ is the Gaussian defined in \eqref{eq: limit of mu and pi^0} and $p^*(\mathbb{G}_{\bm{\mu}, \pi^0})$ is the set of optimal solutions to the linear program:
\begin{equation}\label{eq: auxiliary LP}
    \min_{p \in \mathbb{R}^{S_*}} \langle \boldsymbol{c}, p \rangle,\quad \textrm{s.t.}\ \bm{A}p= \mathbb{G}_{\bm{\mu}, \pi^0},\ p_{ij} \geq 0 \text{ for all } (ij) \notin \text{spt}(\gamma(\bm{b})).
\end{equation}
This follows from \Cref{lemma: aymptotic of linear programming} (\cite[Theorem 6.1]{limit_random_lp22} and \cite[Theorem 4]{liu2023asymptotic}). Note that $\gamma(\bm{b}) =\{ (\mathrm{Id}, \mathrm{Id})_\#(\pi^0) \}$ since $c=0$ only on the diagonal. Recalling the fact that $\hat{\pi}(\bm{b}_{n,m}) = (\mathrm{P}_{\boldsymbol{x}})_\# (\gamma(\bm{b}_{n,m}))$, applying 
$(\mathrm{P}_{\boldsymbol{x}})_\#$ to \eqref{eqn:limit-coupling} yields the conclusion by continuous mapping theorem.
\end{proof}
\end{theorem}

\begin{remark}[Asymptotic Behavior and Efficiency]
The convergence rate in \cref{thm: limit distribution of projection}, that is, $\sqrt{nm/(n+m)}$ is parametric (root-$n$ type) in the combined sample sizes. When $n \gg m$ as in the ACS data, this rate simplifies to approximately $\sqrt{m}$, indicating that the smaller sample governs the effective rate. The main practical benefit of incorporating marginal data is not an increase in the convergence rate itself, but that the estimator $\hat{\pi}$ provides a more accurate point estimate of $\pi^0$ than the empirical measure based only on the coupled data, $\pi_m^0$. This improvement is reflected in the narrower confidence intervals shown in Sections \ref{sec:sim} and \ref{sec:data}.
\end{remark}

For statistical inference tasks, one needs more than \Cref{thm: limit distribution of projection} since it does not provide any
data-driven description of asymptotically valid confidence sets. Two approaches can be used to obtain such confidence sets.  The first approach is subsampling \cite{politis1994large, politis1999subsampling}, also known as the $m$ out of $n$ bootstrap. The validity of subsampling relies on the existence of an asymptotic distribution for the estimator, which is ensured by \Cref{thm: limit distribution of projection}. The procedure is implemented in available software, such as the \texttt{moonboot} R package \citep{moonboot,Dalitz_2025}, and we refer readers to these references for implementation details.

The second approach follows \cite{liu2023asymptotic}, in which asymptotic confidence sets are constructed using the auxiliary linear program \eqref{eq: auxiliary LP}. Although their result rigorously guarantees an at least $1-\alpha$ confidence set, it is too conservative to apply it directly to real data,  compared to the subsampling method. The full statement of this approach is presented in \Cref{subsec: confidence sets}.



\begin{remark}
While the bootstrap \citep{efron1994introduction} is widely used and often viewed as a convenient tool for practical inference, we do not recommend it for inference on the estimator obtained from the linear program \eqref{def: finite linear primal}. The bootstrap is known to be invalid for non-regular estimators \citep{shao1994bootstrap, andrews2000inconsistency, fang2019inference}, and the estimator derived from \eqref{def: finite linear primal} is generally non-regular. As a result, bootstrap procedures may fail to yield valid statistical inference. Our simulation results indeed confirm that the standard bootstrap can produce invalid confidence intervals; see \Cref{sec:sim}. 
\end{remark}

\section{Simulation} \label{sec:sim}

We conducted simulation studies to investigate the finite sample performance of the proposed approach.

\subsection{Infinite Support Case}
First, we considered a setting with infinite support, i.e., continuous random variables. The number of coupled observations was varied over $m \in \{ 50, 100, 200 \}$ and the number of decoupled observations was varied over $n \in \{5m, 10m, 25m \}$. We then generated $(m+n)$  i.i.d. coupled data from the bivariate normal distribution as follows:
\begin{align*}
    (Z_1,Z_2) \sim 
    \pi^0 =
    N 
    \left(
    \begin{pmatrix}
        0 \\ 0
    \end{pmatrix},
    \begin{pmatrix}
        1 & \rho \\ \rho & 1
    \end{pmatrix}
    \right) \ .
\end{align*}
The covariance parameter $\rho$ was varied over $ \rho \in \{0, 0.25, 0.75\}$ to examine the performance of the proposed method under conditions ranging from no correlation to a strong association between $Z_1$ and $Z_2$. The marginal distributions of $Z_1$ and $Z_2$ are the same as $\mu=\nu=N(0,1)$. The coupled dataset consisted of the first $m$ $(Z_1,Z_2)$ pairs, while the decoupled dataset is obtained by de-coupling the remaining $n$ pairs and randomly permuting them.

We estimated the coupling for the marginal data from \eqref{eq: entropic MOT formulation of Wasserstein projection} with the squared Euclidean cost. For the entropic regularization parameter $\eta$, we considered three settings: (i) fixed values $\eta \in \{2^{0},2^{3},2^{-6} \}$; and (ii) a data-driven choice of $\eta$ based on a cross-validation procedure, which we denote by $\eta_{\text{cv}}$; the details of this procedure are provided in Algorithm~\ref{alg:cv eta}. Roughly speaking, the cross-validation procedure consists of the following steps:
(i) partition the coupled dataset into non-overlapping, approximately equal-sized $K$ folds; (ii) apply the proposed extension of coupling framework to the dataset excluding one fold;
(iii) evaluate the cost on the held-out fold; (iv) repeat steps (ii) and (iii) for each of the $K$ folds. Finally, $\eta_{\text{cv}}$ is defined as the value of $\eta$ that minimizes the average cost across all held-out folds. To reduce the dependence on a particular split, steps (i)-(iv) can be repeated $R$ times. For the simulation study, we used $R = 10$ repetitions and $K = 5$ folds, with the candidate set of $\eta$ values given by $\mathcal{E} = \{ 2^{0}, 2^{-1}, 2^{-2},2^{-3},2^{-4}, 2^{-5}, 2^{-6} \}$.

\begin{algorithm}[!htp]
    \begin{algorithmic}[1]    
        \STATE Define a set of candidate values for $\eta$, denoted by $\mathcal{E} = \{ \eta_1,\ldots,\eta_J \}$
        \FOR{$j=1,\ldots,J$}
        \STATE Let $R$ be the number of cross-validation repetitions
        \FOR{$r=1,\ldots,R$}
        \STATE Partition $\mathcal{I} = \{1, \ldots, m\}$ into $K$ non-overlapping, approximately equal-sized folds, denoted by $\mathcal{I}_{r,1}, \ldots, \mathcal{I}_{r,K}$
        \FOR{$k=1,\ldots,K$}
        \STATE Let $\widehat{\pi}_{r,k}^{(j)}$ denote an estimate of $\pi^0$ obtained by using $\mathcal{I} \setminus \mathcal{I}_{r,k}$ as the coupled dataset and $\mathcal{I}_{r,k}$ as the marginal dataset, with $\eta_j$ as the entropic regularization parameter
        \STATE Evaluate the cost over $\mathcal{I}_{r,k}$ using the estimated coupling:
        \begin{align*}
            \mathcal{C}_{r,k}^{(j)}
            =
            \frac{1}{|\mathcal{I}_{r,k}|}
            \sum_{i \in \mathcal{I}_{r,k}}
            \int c((x,y),(x_i,y_i)) \, d\widehat{\pi}_{r,k}^{(j)} (x,y)
        \end{align*}
        \ENDFOR %
        \ENDFOR %
        \ENDFOR %
        \STATE Select the value of $\eta$ that minimizes the average cost across folds and repetitions:
\begin{align*}
    \eta_{\text{cv}}
    =
    \eta_{j^*}
    \quad \text{ where }
    j^*
    &= \argmin_{j \in \{ 1, \ldots, J\} }
       \frac{1}{R} 
       \sum_{r=1}^{R} \sum_{k=1}^{K} \mathcal{C}_{r,k}^{(j)},
\end{align*}
    \end{algorithmic}    
    \caption{Repeated Cross-validation for Selecting $\eta$}
    \label{alg:cv eta}
\end{algorithm}

To assess the performance of the estimated coupling $\widehat{\pi}$, we computed estimates of the covariance and the cumulative distribution function at $(-0.25, -0.25)$, given by
\begin{align*}
&
\widehat{\rho}
= \int z_1 z_2 \, d\widehat{\pi}(z_1,z_2)
- \int z_1 \, d\widehat{\pi}(z_1,z_2) \cdot \int z_2 \, d\widehat{\pi}(z_1,z_2), \\
&
\widehat{F}(-0.25, -0.25)
= \int \mathbbm{1}(z_1 \leq -0.25, z_2 \leq -0.25) \,  d\widehat{\pi}(z_1,z_2).
\end{align*}
If the estimated coupling $\widehat{\pi}$ is close to the true coupling $\pi^0$, then $\widehat{\rho}$ and $\widehat{F}(-0.25, -0.25)$ will also be close to $\rho$ and $F(-0.25, -0.25) = \int \mathbbm{1}(z_1 \leq -0.25, z_2 \leq -0.25) \, d\pi^0(z_1,z_2)$. For comparison, we also computed the same quantities based on the $m$ coupled data, denoted by $\widetilde{\rho}$ and $\widetilde{F}(-0.25,-0.25)$; note that these estimators are consistent by the law of large numbers.

Table \ref{tab:Sim1} presents the root mean squared error (RMSE) of $\widetilde{\rho}$ and $\widehat{\rho}$, averaged over 500 repetitions, for each combination of $(m, n, \rho)$. First, for each $\rho$ and fixed ratio $n/m$, the RMSE of $\widetilde{\rho}$ decreases with $m$, approximately at the rate of $1/\sqrt{m}$, which corresponds to the parametric rate. Similarly, the RMSE of $\widehat{\rho}$ also decreases with $m$ at roughly the same rate.

Second, the performance of $\widehat{\rho}$ depends on both the entropic regularization parameter $\eta$ and the true value of $\rho$. Specifically, when $\rho = 0$ (i.e., $z_1$ and $z_2$ are independent), larger values of $\eta$ yield the smallest RMSE among the candidate choices. Conversely, when $\rho = 0.75$ (i.e., $z_1$ and $z_2$ are strongly dependent), smaller values of $\eta$ achieve the lowest RMSE. Moreover, if $\eta$ is mismatched to $\rho$---for example, a small $\eta$ when $\rho$ is small or a large $\eta$ when $\rho$ is large---the resulting $\widehat{\rho}$ performs worse than $\widetilde{\rho}$. This indicates that, under such settings, the proposed optimal extension of coupling method offers no practical advantage, as the estimator based solely on the coupled data performs better. 

Third, regardless of the value of $\rho$, the cross-validation procedure yields $\widehat{\rho}$ that outperforms $\widetilde{\rho}$. This shows that selecting $\eta$ via cross-validation consistently enables the use of marginal data to obtain an estimator more efficient than one based solely on the coupled data. As noted above, such uniform performance gains are not achievable with a fixed, a priori choice of $\eta$. Moreover, the last column of Table \ref{tab:Sim1} shows that the procedure tends to select a large $\eta$ when $\rho$ is small and a small $\eta$ when $\rho$ is large, indicating that cross-validation automatically adapts to the underlying dependence structure.

\begin{table}[!htp]
\scriptsize
\renewcommand{\arraystretch}{1.2}
\centering
\setlength{\tabcolsep}{6pt}
\begin{tabular}{|c|c|c|cccc|c|}
\hline
\multirow{3}{*}{$\rho$} & \multirow{3}{*}{$n/m$} & \multirow{3}{*}{$m$} & \multicolumn{4}{c|}{RMSE}                                                                                                                                 & \multirow{3}{*}{Average $\log_2 \eta_{\text{cv}}$} \\ \cline{4-7}
                        &                        &                      & \multicolumn{1}{c|}{\multirow{2}{*}{$\widetilde{\rho}$}} & \multicolumn{3}{c|}{$\widehat{\rho}$}                                                          &                                                    \\ \cline{5-7}
                        &                        &                      & \multicolumn{1}{c|}{}                                    & \multicolumn{1}{c|}{$\eta=2^0$} & \multicolumn{1}{c|}{$\eta=2^{-6}$} & $\eta=\eta_{\text{cv}}$ &                                                    \\ \hline
\multirow{9}{*}{0.00}   & \multirow{3}{*}{5}     & 50                   & \multicolumn{1}{c|}{14.19}                               & \multicolumn{1}{c|}{8.71}       & \multicolumn{1}{c|}{14.34}         & 13.00                   & -2.63                                              \\ \cline{3-8} 
                        &                        & 100                  & \multicolumn{1}{c|}{10.02}                               & \multicolumn{1}{c|}{6.13}       & \multicolumn{1}{c|}{10.01}         & 9.05                    & -2.75                                              \\ \cline{3-8} 
                        &                        & 200                  & \multicolumn{1}{c|}{6.89}                                & \multicolumn{1}{c|}{4.22}       & \multicolumn{1}{c|}{6.93}          & 6.19                    & -2.71                                              \\ \cline{2-8} 
                        & \multirow{3}{*}{10}    & 50                   & \multicolumn{1}{c|}{14.19}                               & \multicolumn{1}{c|}{8.69}       & \multicolumn{1}{c|}{14.35}         & 12.96                   & -2.53                                              \\ \cline{3-8} 
                        &                        & 100                  & \multicolumn{1}{c|}{10.02}                               & \multicolumn{1}{c|}{6.11}       & \multicolumn{1}{c|}{10.00}         & 9.18                    & -3.08                                              \\ \cline{3-8} 
                        &                        & 200                  & \multicolumn{1}{c|}{6.89}                                & \multicolumn{1}{c|}{4.21}       & \multicolumn{1}{c|}{6.93}          & 6.17                    & -2.70                                              \\ \cline{2-8} 
                        & \multirow{3}{*}{25}    & 50                   & \multicolumn{1}{c|}{14.19}                               & \multicolumn{1}{c|}{8.66}       & \multicolumn{1}{c|}{14.29}         & 13.02                   & -2.63                                              \\ \cline{3-8} 
                        &                        & 100                  & \multicolumn{1}{c|}{10.02}                               & \multicolumn{1}{c|}{6.11}       & \multicolumn{1}{c|}{10.00}         & 9.06                    & -2.83                                              \\ \cline{3-8} 
                        &                        & 200                  & \multicolumn{1}{c|}{6.89}                                & \multicolumn{1}{c|}{4.21}       & \multicolumn{1}{c|}{6.94}          & 6.19                    & -2.70                                              \\ \hline
\multirow{9}{*}{0.25}   & \multirow{3}{*}{5}     & 50                   & \multicolumn{1}{c|}{14.46}                               & \multicolumn{1}{c|}{12.88}      & \multicolumn{1}{c|}{13.47}         & 13.77                   & -4.03                                              \\ \cline{3-8} 
                        &                        & 100                  & \multicolumn{1}{c|}{10.28}                               & \multicolumn{1}{c|}{11.76}      & \multicolumn{1}{c|}{9.46}          & 10.04                   & -4.49                                              \\ \cline{3-8} 
                        &                        & 200                  & \multicolumn{1}{c|}{7.11}                                & \multicolumn{1}{c|}{10.49}      & \multicolumn{1}{c|}{6.58}          & 6.83                    & -5.21                                              \\ \cline{2-8} 
                        & \multirow{3}{*}{10}    & 50                   & \multicolumn{1}{c|}{14.46}                               & \multicolumn{1}{c|}{12.79}      & \multicolumn{1}{c|}{13.35}         & 13.71                   & -3.79                                              \\ \cline{3-8} 
                        &                        & 100                  & \multicolumn{1}{c|}{10.28}                               & \multicolumn{1}{c|}{11.68}      & \multicolumn{1}{c|}{9.40}          & 10.01                   & -4.51                                              \\ \cline{3-8} 
                        &                        & 200                  & \multicolumn{1}{c|}{7.11}                                & \multicolumn{1}{c|}{10.47}      & \multicolumn{1}{c|}{6.53}          & 6.82                    & -5.22                                              \\ \cline{2-8} 
                        & \multirow{3}{*}{25}    & 50                   & \multicolumn{1}{c|}{14.46}                               & \multicolumn{1}{c|}{12.78}      & \multicolumn{1}{c|}{13.41}         & 13.72                   & -3.82                                              \\ \cline{3-8} 
                        &                        & 100                  & \multicolumn{1}{c|}{10.28}                               & \multicolumn{1}{c|}{11.67}      & \multicolumn{1}{c|}{9.35}          & 9.88                    & -4.67                                              \\ \cline{3-8} 
                        &                        & 200                  & \multicolumn{1}{c|}{7.11}                                & \multicolumn{1}{c|}{10.47}      & \multicolumn{1}{c|}{6.54}          & 6.78                    & -5.14                                              \\ \hline
\multirow{9}{*}{0.75}   & \multirow{3}{*}{5}     & 50                   & \multicolumn{1}{c|}{17.25}                               & \multicolumn{1}{c|}{31.27}      & \multicolumn{1}{c|}{9.09}          & 9.15                    & -5.77                                              \\ \cline{3-8} 
                        &                        & 100                  & \multicolumn{1}{c|}{12.61}                               & \multicolumn{1}{c|}{30.63}      & \multicolumn{1}{c|}{6.52}          & 6.53                    & -5.97                                              \\ \cline{3-8} 
                        &                        & 200                  & \multicolumn{1}{c|}{8.72}                                & \multicolumn{1}{c|}{29.77}      & \multicolumn{1}{c|}{4.32}          & 4.32                    & -6.00                                              \\ \cline{2-8} 
                        & \multirow{3}{*}{10}    & 50                   & \multicolumn{1}{c|}{17.25}                               & \multicolumn{1}{c|}{31.05}      & \multicolumn{1}{c|}{8.00}          & 8.09                    & -5.73                                              \\ \cline{3-8} 
                        &                        & 100                  & \multicolumn{1}{c|}{12.61}                               & \multicolumn{1}{c|}{30.36}      & \multicolumn{1}{c|}{5.66}          & 5.66                    & -5.99                                              \\ \cline{3-8} 
                        &                        & 200                  & \multicolumn{1}{c|}{8.72}                                & \multicolumn{1}{c|}{29.74}      & \multicolumn{1}{c|}{3.91}          & 3.91                    & -6.00                                              \\ \cline{2-8} 
                        & \multirow{3}{*}{25}    & 50                   & \multicolumn{1}{c|}{17.25}                               & \multicolumn{1}{c|}{30.89}      & \multicolumn{1}{c|}{7.36}          & 7.44                    & -5.79                                              \\ \cline{3-8} 
                        &                        & 100                  & \multicolumn{1}{c|}{12.61}                               & \multicolumn{1}{c|}{30.41}      & \multicolumn{1}{c|}{5.22}          & 5.23                    & -5.99                                              \\ \cline{3-8} 
                        &                        & 200                  & \multicolumn{1}{c|}{8.72}                                & \multicolumn{1}{c|}{29.73}      & \multicolumn{1}{c|}{3.61}          & 3.61                    & -6.00                                              \\ \hline
\end{tabular}
\caption{\scriptsize Summary of simulation results under the infinite support case. The four subcolumns under ``RMSE'' report the RMSE of the estimators for $\rho$, computed over 500 simulation repetitions. All RMSE values are scaled by a factor of 100. The last column, labeled ``Average $\log_{2} \eta_{\text{cv}}$,'' reports the mean of $\log_{2}\eta_{\text{cv}}$ across 500 repetitions.}
\label{tab:Sim1}
\end{table}

In Appendix \ref{sec:appendix:additional sim}, we present the results for $\widehat{F}(-0.25, -0.25)$, which exhibit patterns similar to those in Table \ref{tab:Sim1}. These simulation results indicate that the proposed method is consistent with the properties described in \Cref{cor: consistency of entropic estimator} and demonstrates practical merit when the entropic regularization parameter $\eta$ is selected in a data-driven manner.

\subsection{Finite Support Case}

Next, we examined a setting with finite support. In this setting, we fixed the size of the decoupled dataset at $n = 10^6$ and varied the number of coupled observations over $m \in \{1000, 2000, 4000, 8000, 16000, 32000\}$. We generated $m+n$ i.i.d. observations of the discrete variables $(Z_1, Z_2, Z_3) \in \{1,2\} \otimes \{1,2,3,4\} \otimes \{1,2,3,4,5\}$ according to the following joint distribution: 
\begin{align*}
    \Pr(Z_1=z_1, Z_2=z_2, Z_3=z_3)
    =
    \left\{
    \begin{array}{ll}
    7/105 & \text{ if } z_1 = 2, z_2=4, z_3 \in \{1,2,3,4\} 
    \\
    42/105 & \text{ if } z_1 = 2, z_2=4, z_3 =5 
    \\
    1/105 & \text{ otherwise}
    \end{array}
    \right. 
\end{align*}
The coupled dataset was formed from the first $m$ observations, while the decoupled dataset was generated by decoupling the remaining $n$ observations.

Given the finite support of the data, we estimated the coupling for the marginal data using the method described in Section \ref{section: finite support case}. We employed a separable cost function defined as:
\begin{align*}
    c((z_1,z_2,z_3),(z_1',z_2',z_3'))
    =
    \sum_{j=1}^{3} | z_j-z_j' | \ .
\end{align*}
To assess the accuracy of the estimated coupling $\widehat{\pi}$, we estimated the conditional probability $\psi^0 = \Pr(Z_3=5 \mid Z_1=2, Z_2=4) = 0.6$. If $\widehat{\pi}$ is close to the true coupling $\pi^0$, the resulting estimate $\widehat{\psi}$ should also be close to $\psi^0$. For comparison, we also computed the estimator $\widetilde{\psi} = \sum_{i=1}^{m} \mathbb{I} (Z_{i,1}=2,Z_{i,2}=4,X_{i,3}=5) / \sum_{i=1}^{m} \mathbb{I}(Z_{i,1}=2,Z_{i,2}=4)$ where $(Z_{i,1}, Z_{i,2}, Z_{i,3})$ denotes the $i$th coupled observation. Note that $\widetilde{\psi}$ is the maximum likelihood estimator of $\psi^0$ based on the coupled data and is a consistent estimator of $\psi^0$. For inference, we compare 95\% confidence intervals for $\psi^0$ constructed using the following methods:
\begin{itemize}
\item $\text{CI}(\widetilde{\psi},\text{boot})$: a 95\% confidence interval based on $\widetilde{\psi}$ obtained via the nonparametric bootstrap.

\item $\text{CI}(\widehat{\psi},\text{boot})$: a 95\% confidence interval based on $\widehat{\psi}$ obtained via the nonparametric bootstrap.

\item $\text{CI}(\widehat{\psi},\text{ss})$: a 95\% confidence interval based on $\widehat{\psi}$ constructed using subsampling.

\item $\text{CI}(\widehat{\psi},\text{Liu})$: a 95\% confidence interval based on $\widehat{\psi}$ constructed using the method of \citet{liu2023asymptotic}.

\end{itemize}

Table \ref{tab:Sim Discrete} reports the empirical bias and root mean squared error (RMSE) of $\widetilde{\psi}$ and $\widehat{\psi}$, as well as the coverage probabilities and average lengths of the associated confidence intervals, based on 1000 repetitions. The estimator $\widetilde{\psi}$ exhibits nearly negligible bias across all sample sizes, whereas $\widehat{\psi}$ shows a slightly larger, yet modest, finite-sample bias that diminishes as $M$ increases. Despite this slight bias, $\widehat{\psi}$ demonstrates notable gains in efficiency, yielding a consistently lower RMSE than $\widetilde{\psi}$ for all $M$.

Regarding inference, the bootstrap confidence interval for $\widetilde{\psi}$ achieves empirical coverage rates very close to the nominal 95\% level across all values of $m$. This is expected, as the bootstrap is a valid method for $\widetilde{\psi}$, which is the maximum likelihood estimator for $\psi^0$ based on the coupled data. For $\widehat{\psi}$, the standard bootstrap interval tends to undercover (ranging around 88\% to 90\%). This aligns with our discussion in Section \ref{section: finite support case}, where we noted that the bootstrap is generally an invalid inference method for our estimator. Conversely, the confidence intervals based on subsampling and the approach of \citet{liu2023asymptotic} are valid, as they successfully attain the nominal coverage. Between these two, the subsampling confidence interval performs better, with empirical coverage approaching the nominal level as $m$ increases, while the method of \citet{liu2023asymptotic} tends to be conservative, consistently overcovering (around 99\%). We remark that the lengths of these valid intervals for $\widehat{\psi}$ are shorter than $\text{CI}(\widetilde{\psi},\text{boot})$, indicating that our approach enables more efficient inference. 

Overall, the simulation results demonstrate the promise of our approach: while $\widetilde{\psi}$, which relies only on the coupled data, provides minimal bias and valid inference via the bootstrap, $\widehat{\psi}$---which further incorporates marginal data---offers more efficient inference for the underlying distribution.

\begin{table}[!htp]
\scriptsize
\renewcommand{\arraystretch}{1.2}
\centering
\setlength{\tabcolsep}{2pt}
\begin{tabular}{|c|cc|cc|cccc|cccc|}
\hline
\multirow{2}{*}{$m$\,($\times 10^3$)} & \multicolumn{2}{c|}{Bias\,($\times 10^3$)} & \multicolumn{2}{c|}{RMSE\,($\times 10^3$)} & \multicolumn{4}{c|}{CI Coverage\,(\%)} & \multicolumn{4}{c|}{CI Length\,($\times 10^3$)} \\ \cline{2-13} 
 & \multicolumn{1}{c|}{$\widetilde{\psi}$} & $\widehat{\psi}$ & \multicolumn{1}{c|}{$\widetilde{\psi}$} & $\widehat{\psi}$ & \multicolumn{1}{c|}{$\widetilde{\psi},\text{boot}$} & \multicolumn{1}{c|}{$\widehat{\psi},\text{boot}$} & \multicolumn{1}{c|}{$\widehat{\psi},\text{SS}$} & $\widehat{\psi},\text{Liu}$ & \multicolumn{1}{c|}{$\widetilde{\psi},\text{boot}$} & \multicolumn{1}{c|}{$\widehat{\psi},\text{boot}$} & \multicolumn{1}{c|}{$\widehat{\psi},\text{SS}$} & $\widehat{\psi},\text{Liu}$ \\ \hline
1 & \multicolumn{1}{c|}{-1.1} & -6.2 & \multicolumn{1}{c|}{19.1} & 12.1 & \multicolumn{1}{c|}{94.7} & \multicolumn{1}{c|}{88.4} & \multicolumn{1}{c|}{91.5} & 98.7 & \multicolumn{1}{c|}{74.4} & \multicolumn{1}{c|}{40.5} & \multicolumn{1}{c|}{39.0} & 66.2 \\ \hline
2 & \multicolumn{1}{c|}{-0.8} & -4.3 & \multicolumn{1}{c|}{13.5} & 8.5 & \multicolumn{1}{c|}{94.3} & \multicolumn{1}{c|}{89.5} & \multicolumn{1}{c|}{92.6} & 98.9 & \multicolumn{1}{c|}{52.6} & \multicolumn{1}{c|}{29.1} & \multicolumn{1}{c|}{28.6} & 46.7 \\ \hline
4 & \multicolumn{1}{c|}{-0.4} & -2.9 & \multicolumn{1}{c|}{9.3} & 5.9 & \multicolumn{1}{c|}{96.1} & \multicolumn{1}{c|}{90.1} & \multicolumn{1}{c|}{94.1} & 99.5 & \multicolumn{1}{c|}{37.2} & \multicolumn{1}{c|}{20.8} & \multicolumn{1}{c|}{20.3} & 33.0 \\ \hline
8 & \multicolumn{1}{c|}{-0.2} & -2.0 & \multicolumn{1}{c|}{6.7} & 4.2 & \multicolumn{1}{c|}{95.4} & \multicolumn{1}{c|}{90.1} & \multicolumn{1}{c|}{94.8} & 98.8 & \multicolumn{1}{c|}{26.3} & \multicolumn{1}{c|}{14.8} & \multicolumn{1}{c|}{14.9} & 23.6 \\ \hline
16 & \multicolumn{1}{c|}{-0.3} & -1.5 & \multicolumn{1}{c|}{4.8} & 2.9 & \multicolumn{1}{c|}{94.9} & \multicolumn{1}{c|}{89.4} & \multicolumn{1}{c|}{96.5} & 99.3 & \multicolumn{1}{c|}{18.6} & \multicolumn{1}{c|}{10.5} & \multicolumn{1}{c|}{10.7} & 16.9 \\ \hline
32 & \multicolumn{1}{c|}{-0.1} & -1.1 & \multicolumn{1}{c|}{3.4} & 2.2 & \multicolumn{1}{c|}{94.8} & \multicolumn{1}{c|}{88.9} & \multicolumn{1}{c|}{94.3} & 98.9 & \multicolumn{1}{c|}{13.1} & \multicolumn{1}{c|}{7.5} & \multicolumn{1}{c|}{7.7} & 12.2 \\ \hline
\end{tabular}
\caption{\scriptsize Summary of simulation results under the finite support case. All bias, RMSE, and CI length values are scaled by a factor of 1000. CI coverage is reported as a percentage.}
\label{tab:Sim Discrete}
\end{table}

\section{Real-world Application} \label{sec:data}

We applied our proposed method to a real-world survey dataset. Consistent with \Cref{sec: intro}, we focused on individuals in Illinois, using the U.S. Census Bureau’s ACS 2018 datasets for both the coupled and marginal data. Specifically, the coupled dataset consisted of the ACS 2018 5-year PUMS with $m \simeq 6.13 \times 10^5$, while the marginal dataset was drawn from the ACS 2018 5-year Summary File, representing the Illinois population with a total size of $n \approx 1.28 \times 10^7$. 

For these datasets, we considered five discrete variables $(Z_1,Z_2,Z_3,Z_4,Z_5) \in \{0,1\} \otimes \{0,1\} \otimes \{0,1\} \otimes \{0,1,2,3\} \otimes \{0,1\}$, defined as:
\begin{itemize}
\item $Z_1 = \mathbb{I}(\text{Individual has health insurance})$
\item $Z_2 = \mathbb{I}(\text{Age} \in  [25,64]\text{, i.e., working age})$
\item $Z_3 = \mathbb{I}(\text{Sex = Male})$
\item $Z_4 = \mathbb{I}(\text{Race = White}) + 2 \mathbb{I}(\text{Race = Black}) + 3 \mathbb{I}(\text{Race = Asian})$
\item $Z_5 = \mathbb{I}(\text{Education $\leq$ High school})$
\end{itemize}

Using these variables, we examined the relationship between health insurance status and socioeconomic factors, defined as $\psi^0_{ijkl} = \Pr(Z_1=0 \mid Z_2=i, Z_3=j, Z_4=k, Z_5=l)$, the conditional probability of not having health insurance given socioeconomic factors. Following the simulation study in Section \ref{sec:sim}, we compared two estimators: $\widehat{\psi}_{ijkl}$, obtained from our proposed method in Section \ref{section: finite support case}, and $\widetilde{\psi}_{ijkl}$, the maximum likelihood estimator based on the coupled data, and evaluated four 95\% confidence intervals for inference.  

Tables \ref{tab:data} and \ref{tab:data2} present the estimation results for $\psi_{00kl}^{0}$, representing the conditional probabilities of not having health insurance (scaled as percentages) for females outside the working-age range in Illinois, stratified by race and educational attainment. Consistent with the simulation results in Section \ref{sec:sim}, incorporating the marginal data via our proposed estimator $\widehat{\psi}$ yields noticeable differences compared to the standard maximum likelihood estimator $\widetilde{\psi}$ derived solely from the coupled data. While the overall values between the two estimators are similar, $\widehat{\psi}$ produces slightly higher estimates for most demographic groups (White, Black, and Others) and slightly lower estimates for the Asian demographic. 

In terms of confidence intervals, $\widehat{\psi}$ generally delivers tighter confidence intervals under both bootstrap and subsampling procedures compared with bootstrap intervals based on $\widetilde{\psi}$. However, although the 95\% bootstrap confidence intervals based on $\widehat{\psi}$ are substantially narrower than the alternatives, simulation results indicate that they may fail to attain the nominal 95\% coverage level. Moreover, while the confidence intervals proposed by \citet{liu2023asymptotic} are theoretically valid, they are typically wider than the bootstrap intervals based on $\widetilde{\psi}$, resulting in less efficient inference. Consistent with the simulation findings, we therefore recommend subsampling confidence intervals for inference. 

Comparing the lengths of the 95\% confidence intervals from subsampling based on $\widehat{\psi}$ with the bootstrap intervals based on $\widetilde{\psi}$ highlights the efficiency gains achieved by incorporating marginal information. For example, for $\psi_{0011}^0$ (the first row in Tables \ref{tab:data} and \ref{tab:data2}), the corresponding interval widths are 0.26 and 0.31, respectively. This implies that, when an investigator relies solely on coupled data and conducts inference using $\widetilde{\psi}$, approximately $(0.31/0.26)^2 \simeq 1.41$ times more coupled data is required to achieve the level of precision attainable when incorporating marginal data through $\widehat{\psi}$. Thus, leveraging marginal data enables efficient statistical inference without requiring substantially more coupled data to achieve the same level of precision.

Overall, these results reinforce a key takeaway for real-world applications: while incorporating decoupled marginal data enables the construction of improved OT-based estimators, the choice of inference procedure remains critical for balancing efficiency with statistical validity.

\begin{table}[!htp]
\scriptsize
\renewcommand{\arraystretch}{1.2}
\centering
\setlength{\tabcolsep}{2pt}

\begin{tabular}{|c|c|cc|}
\hline
\multirow{2}{*}{Race} & \multirow{2}{*}{\begin{tabular}[c]{@{}c@{}}Education\\ $\leq$ HS?\end{tabular}} & \multicolumn{2}{c|}{$\widetilde{\psi}$} \\ \cline{3-4} 
 & & \multicolumn{1}{c|}{Estimate(\%)} & 95\% Bootstrap CI \\ \hline
\multirow{2}{*}{White} & Yes & \multicolumn{1}{c|}{$2.98$} & $(2.83,\,3.14)\,[0.31]$ \\ \cline{2-4} 
 & No & \multicolumn{1}{c|}{$3.72$} & $(3.33,\,4.13)\,[0.80]$ \\ \hline
\multirow{2}{*}{Black} & Yes & \multicolumn{1}{c|}{$5.89$} & $(4.96,\,6.78)\,[1.82]$ \\ \cline{2-4} 
 & No & \multicolumn{1}{c|}{$6.40$} & $(5.71,\,7.11)\,[1.41]$ \\ \hline
\multirow{2}{*}{Asian} & Yes & \multicolumn{1}{c|}{$1.77$} & $(1.58,\,1.98)\,[0.40]$ \\ \cline{2-4} 
 & No & \multicolumn{1}{c|}{$2.80$} & $(2.12,\,3.53)\,[1.40]$ \\ \hline
\multirow{2}{*}{Others} & Yes & \multicolumn{1}{c|}{$5.04$} & $(3.90,\,6.21)\,[2.31]$ \\ \cline{2-4} 
 & No & \multicolumn{1}{c|}{$3.78$} & $(1.93,\,6.10)\,[4.17]$ \\ \hline
\end{tabular}
\caption{\scriptsize Summary of the real-world application for females outside the working-age group based on $\widetilde{\psi}$. All values are scaled by a factor of 100. Numbers in square brackets indicate the lengths of the corresponding confidence intervals.}
\label{tab:data}
\end{table}

\begin{table}[!htp]
\scriptsize
\renewcommand{\arraystretch}{1.2}
\centering
\setlength{\tabcolsep}{2pt}
\begin{tabular}{|c|c|cccc|}
\hline
\multirow{2}{*}{Race} & \multirow{2}{*}{\begin{tabular}[c]{@{}c@{}}Education\\ $\leq$ HS?\end{tabular}} & \multicolumn{4}{c|}{$\widehat{\psi}$} \\ \cline{3-6} 
 & & \multicolumn{1}{c|}{Estimate(\%)} & \multicolumn{1}{c|}{95\% Bootstrap CI} & \multicolumn{1}{c|}{95\% Subsampling CI} & 95\% CI of \citet{liu2023asymptotic} \\ \hline
\multirow{2}{*}{White} & Yes & \multicolumn{1}{c|}{$3.01$} & \multicolumn{1}{c|}{$(2.88,\,3.15)\,[0.27]$} & \multicolumn{1}{c|}{$(2.85,\,3.11)\,[0.26]$} & $(2.84,\,3.18)\,[0.34]$ \\ \cline{2-6} 
 & No & \multicolumn{1}{c|}{$3.80$} & \multicolumn{1}{c|}{$(3.45,\,4.16)\,[0.71]$} & \multicolumn{1}{c|}{$(3.32,\,4.08)\,[0.76]$} & $(3.42,\,4.19)\,[0.78]$ \\ \hline
\multirow{2}{*}{Black} & Yes & \multicolumn{1}{c|}{$5.73$} & \multicolumn{1}{c|}{$(4.90,\,6.55)\,[1.65]$} & \multicolumn{1}{c|}{$(4.55,\,6.38)\,[1.83]$} & $(4.79,\,6.74)\,[1.96]$ \\ \cline{2-6} 
 & No & \multicolumn{1}{c|}{$6.73$} & \multicolumn{1}{c|}{$(6.11,\,7.40)\,[1.29]$} & \multicolumn{1}{c|}{$(5.90,\,7.26)\,[1.36]$} & $(6.03,\,7.45)\,[1.42]$ \\ \hline
\multirow{2}{*}{Asian} & Yes & \multicolumn{1}{c|}{$1.80$} & \multicolumn{1}{c|}{$(1.63,\,1.98)\,[0.35]$} & \multicolumn{1}{c|}{$(1.56,\,1.92)\,[0.37]$} & $(1.57,\,2.03)\,[0.46]$ \\ \cline{2-6} 
 & No & \multicolumn{1}{c|}{$2.91$} & \multicolumn{1}{c|}{$(2.29,\,3.55)\,[1.25]$} & \multicolumn{1}{c|}{$(1.84,\,3.24)\,[1.40]$} & $(2.21,\,3.61)\,[1.41]$ \\ \hline
\multirow{2}{*}{Others} & Yes & \multicolumn{1}{c|}{$4.91$} & \multicolumn{1}{c|}{$(3.93,\,5.90)\,[1.98]$} & \multicolumn{1}{c|}{$(3.47,\,5.51)\,[2.04]$} & $(3.60,\,5.96)\,[2.36]$ \\ \cline{2-6} 
 & No & \multicolumn{1}{c|}{$4.31$} & \multicolumn{1}{c|}{$(3.04,\,5.96)\,[2.92]$} & \multicolumn{1}{c|}{$(1.03,\,5.05)\,[4.02]$} & $(2.73,\,6.23)\,[3.51]$ \\ \hline
\end{tabular}
\caption{\scriptsize Summary of the real-world application for females outside the working-age group based on $\widetilde{\psi}$. All values are scaled by a factor of 100. Numbers in square brackets indicate the lengths of the corresponding confidence intervals.}
\label{tab:data2}
\end{table}

\section{Conclusion}\label{sec: conclusion}
In this work, we study the extension of coupling via the optimal transport projection framework. Our proposed approach provides the most geometrically optimal coupling among all possible ones without any assumptions, which guarantees full generality for any type of data. Compared to previous literature, which does not establish an intuitive interpretation, the proposed estimator admits a very natural geometric interpretation, a projection over the Wasserstein space, which allows us to verify consistency, stability, sample complexity, limit distribution and confidence sets. Furthermore, it enjoys nice computational benefits by shadow and entropic regularization, by which the corresponding optimization is solvable almost linear time. Also, we provide several real/synthetic data experiments that empirically support the strength of the proposed method. A practically important finding is that subsampling confidence intervals based on the OT estimator $\hat{\pi}$ are narrower than bootstrap confidence intervals based on the coupled data-only estimator, as demonstrated in both the simulation and ACS application. This means that by incorporating freely available marginal data through our framework, one can achieve more efficient inference about the joint distribution without the need to collect additional coupled data. We believe that our work opens a new direction for recovering a true coupling distribution in terms of the fully nonparametric perspective.

There are many new follow-up questions of interest to both theorists and practitioners. Let us pose three questions. The first is about the relation of the structural assumption of $\pi^0$ to the estimation. Although our framework works for general $\pi^0$, specialists have their own a priori knowledge about it which is informative for inference. Thus, it is desirable if one can exploit this a priori knowledge to improve the accuracy of the proposed estimator. For instance, if $\pi^0$ is generated by some map $T$ which is smooth and monotone, in other words, if strong correlation exists among variables, is there a better way to incorporate it with the method? In contrast, if $\pi^0$ is close to a product measure, is just the product of empirical measures optimal?

The second is to extend the limit distribution theory to continuous settings. The limit distribution we develop is valid only for finite support cases. There is a fundamental bottleneck for this limitation, which is related to the limit distribution of Wasserstein distances of empirical distributions. That for dimensions greater than or equal to 3 is unknown due to the failure of classical empirical process theory to apply to it. Developing limit distributions beyond finite settings is important for statistical inference, yet adequate mathematical tools for studying them still appear to be lacking.

The last question is about a data-driven choice of the entropic parameter discussed in \Cref{sec:sim}. Our procedure is a variant of cross-validation, and under certain regularity conditions, cross-validation is known to be consistent for selecting the (near) optimal regularization parameter \citep{vdVaart2006, Yang2007CV}. While our approach showed strong performance in simulations, it remains an interesting question whether it can be formally justified from a theoretical standpoint.

\bibliographystyle{plainnat}
\bibliography{references_noURL.bib} 

@article{GOZLAN20173327,
title = {Kantorovich duality for general transport costs and applications},
journal = {Journal of Functional Analysis},
volume = {273},
number = {11},
pages = {3327-3405},
year = {2017},
author = {Nathael Gozlan and Cyril Roberto and Paul-Marie Samson and Prasad Tetali},
}

@article{dudley69,
 author = {R. M. Dudley},
 journal = {The Annals of Mathematical Statistics},
 number = {1},
 pages = {40--50},
 publisher = {Institute of Mathematical Statistics},
 title = {The Speed of Mean Glivenko-Cantelli Convergence},
 volume = {40},
 year = {1969}
}

@article {Dereich_Scheutzow_Schottstedt2013,
    AUTHOR = {Dereich, Steffen and Scheutzow, Michael and Schottstedt, Reik},
     TITLE = {Constructive quantization: approximation by empirical
              measures},
   JOURNAL = {Ann. Inst. Henri Poincar\'{e} Probab. Stat.},
  FJOURNAL = {Annales de l'Institut Henri Poincar\'{e} Probabilit\'{e}s et
              Statistiques},
    VOLUME = {49},
      YEAR = {2013},
    NUMBER = {4},
     PAGES = {1183--1203},
MRREVIEWER = {A.\ F.\ Gualtierotti},
}

@article {JW_FB_sample_rates,
    AUTHOR = {Niles-Weed, Jonathan and Bach, Francis},
     TITLE = {Sharp asymptotic and finite-sample rates of convergence of
              empirical measures in {W}asserstein distance},
   JOURNAL = {Bernoulli},
  FJOURNAL = {Bernoulli. Official Journal of the Bernoulli Society for
              Mathematical Statistics and Probability},
    VOLUME = {25},
      YEAR = {2019},
    NUMBER = {4A},
     PAGES = {2620--2648},
MRREVIEWER = {Aihua\ Xia},
}

@article {NF_AG_rate_Wasserstein,
    AUTHOR = {Fournier, Nicolas and Guillin, Arnaud},
     TITLE = {On the rate of convergence in {W}asserstein distance of the
              empirical measure},
   JOURNAL = {Probab. Theory Related Fields},
  FJOURNAL = {Probability Theory and Related Fields},
    VOLUME = {162},
      YEAR = {2015},
    NUMBER = {3-4},
     PAGES = {707--738},
MRREVIEWER = {Jos\'{e}\ Trashorras},
}

@article {MR4255123,
    AUTHOR = {H\"{u}tter, Jan-Christian and Rigollet, Philippe},
     TITLE = {Minimax estimation of smooth optimal transport maps},
   JOURNAL = {Ann. Statist.},
  FJOURNAL = {The Annals of Statistics},
    VOLUME = {49},
      YEAR = {2021},
    NUMBER = {2},
     PAGES = {1166--1194},
MRREVIEWER = {Marc\ Henry},
}

@article {MR4441130,
    AUTHOR = {Niles-Weed, Jonathan and Berthet, Quentin},
     TITLE = {Minimax estimation of smooth densities in {W}asserstein
              distance},
   JOURNAL = {Ann. Statist.},
  FJOURNAL = {The Annals of Statistics},
    VOLUME = {50},
      YEAR = {2022},
    NUMBER = {3},
     PAGES = {1519--1540},
}

@article{merigot2021non,
  title={Non-asymptotic convergence bounds for Wasserstein approximation using point clouds},
  author={M{\'e}rigot, Quentin and Santambrogio, Filippo and Sarrazin, Cl{\'e}ment},
  journal={Advances in Neural Information Processing Systems},
  volume={34},
  pages={12810--12821},
  year={2021}
}

@article {MR4359822,
    AUTHOR = {Ledoux, Michel and Zhu, Jie-Xiang},
     TITLE = {On optimal matching of {G}aussian samples {III}},
   JOURNAL = {Probab. Math. Statist.},
  FJOURNAL = {Probability and Mathematical Statistics},
    VOLUME = {41},
      YEAR = {2021},
    NUMBER = {2},
     PAGES = {237--265},
}

@article {MR3189084,
    AUTHOR = {Boissard, Emmanuel and Le Gouic, Thibaut},
     TITLE = {On the mean speed of convergence of empirical and occupation
              measures in {W}asserstein distance},
   JOURNAL = {Ann. Inst. Henri Poincar\'{e} Probab. Stat.},
  FJOURNAL = {Annales de l'Institut Henri Poincar\'{e} Probabilit\'{e}s et
              Statistiques},
    VOLUME = {50},
      YEAR = {2014},
    NUMBER = {2},
     PAGES = {539--563},
MRREVIEWER = {Marta\ Tyran-Kami\'{n}ska},
}

@article {MR2861675,
    AUTHOR = {Boissard, Emmanuel},
     TITLE = {Simple bounds for convergence of empirical and occupation
              measures in 1-{W}asserstein distance},
   JOURNAL = {Electron. J. Probab.},
  FJOURNAL = {Electronic Journal of Probability},
    VOLUME = {16},
      YEAR = {2011},
     PAGES = {no. 83, 2296--2333},
MRREVIEWER = {Huse\ Fatki\'{c}},
}

@article {MR2280433,
    AUTHOR = {Bolley, Fran\c{c}ois and Guillin, Arnaud and Villani,
              C\'{e}dric},
     TITLE = {Quantitative concentration inequalities for empirical measures
              on non-compact spaces},
   JOURNAL = {Probab. Theory Related Fields},
  FJOURNAL = {Probability Theory and Related Fields},
    VOLUME = {137},
      YEAR = {2007},
    NUMBER = {3-4},
     PAGES = {541--593},
MRREVIEWER = {Hac\`ene\ Djellout},
}

@article{monge1781memoire,
  title={M{\'e}moire sur la th{\'e}orie des d{\'e}blais et des remblais},
  author={Monge, Gaspard},
  journal={Mem. Math. Phys. Acad. Royale Sci.},
  pages={666--704},
  year={1781}
}

@inproceedings{kantorovich1942translocation,
  title={On the translocation of masses},
  author={Kantorovich, Leonid V},
  booktitle={Dokl. Akad. Nauk. USSR (NS)},
  volume={37},
  pages={199--201},
  year={1942}
}

@inproceedings{kantorovich1948problem,
  title={On a problem of Monge},
  author={Kantorovich, Leonid V},
  booktitle={CR (Doklady) Acad. Sci. URSS (NS)},
  volume={3},
  pages={225--226},
  year={1948}
}

@article{otto2001geometry,
  title={The geometry of dissipative evolution equations: the porous medium equation},
  author={Otto, Felix},
  journal={Comm. Partial Differential Equations},
  volume={26},
  pages={101--174},
  year={2001}
}

@article {Brenier91,
    AUTHOR = {Brenier, Yann},
     TITLE = {Polar factorization and monotone rearrangement of
              vector-valued functions},
   JOURNAL = {Comm. Pure Appl. Math.},
  FJOURNAL = {Communications on Pure and Applied Mathematics},
    VOLUME = {44},
      YEAR = {1991},
    NUMBER = {4},
     PAGES = {375--417},
MRREVIEWER = {Robert\ McOwen},
}

@article {Caffarelli_contraction,
    AUTHOR = {Caffarelli, Luis A.},
     TITLE = {Monotonicity properties of optimal transportation and the
              {FKG} and related inequalities},
   JOURNAL = {Comm. Math. Phys.},
  FJOURNAL = {Communications in Mathematical Physics},
    VOLUME = {214},
      YEAR = {2000},
    NUMBER = {3},
     PAGES = {547--563},
MRREVIEWER = {Ludger\ R\"{u}schendorf},
}

@book {villani2008optimal,
    AUTHOR = {Villani, C\'{e}dric},
     TITLE = {Optimal transport},
    SERIES = {Grundlehren der mathematischen Wissenschaften [Fundamental
              Principles of Mathematical Sciences]},
    VOLUME = {338},
      NOTE = {Old and new},
 PUBLISHER = {Springer-Verlag, Berlin},
      YEAR = {2009},
     PAGES = {xxii+973},
      ISBN = {978-3-540-71049-3},
MRREVIEWER = {Dario\ Cordero-Erausquin},
}

@article{delon2020wasserstein,
  title={A Wasserstein-type distance in the space of Gaussian mixture models},
  author={Delon, Julie and Desolneux, Agnes},
  journal={SIAM Journal on Imaging Sciences},
  volume={13},
  number={2},
  pages={936--970},
  year={2020},
  publisher={SIAM}
}

@inproceedings{cuturi2013sinkhorn,
	author = {Cuturi, Marco},
	booktitle = {Advances in Neural Information Processing Systems},
	editor = {C.J. Burges and L. Bottou and M. Welling and Z. Ghahramani and K.Q. Weinberger},
	pages = {},
	publisher = {Curran Associates, Inc.},
	title = {Sinkhorn Distances: Lightspeed Computation of Optimal Transport},
	volume = {26},
	year = {2013}
}

@article{altschuler2017near,
  title={Near-linear time approximation algorithms for optimal transport via Sinkhorn iteration},
  author={Altschuler, Jason and Niles-Weed, Jonathan and Rigollet, Philippe},
  journal={Advances in neural information processing systems},
  volume={30},
  year={2017}
}

@article{lin2022complexity,
  title={On the complexity of approximating multimarginal optimal transport},
  author={Lin, Tianyi and Ho, Nhat and Cuturi, Marco and Jordan, Michael I},
  journal={Journal of Machine Learning Research},
  volume={23},
  number={65},
  pages={1--43},
  year={2022}
}

@article{Eckstein_Nutz_2022,
author = {Eckstein, Stephan and Nutz, Marcel},
title = {Quantitative Stability of Regularized Optimal Transport and Convergence of Sinkhorn's Algorithm},
journal = {SIAM Journal on Mathematical Analysis},
volume = {54},
number = {6},
pages = {5922-5948},
year = {2022},

eprint = { 
    
        https://doi.org/10.1137/21M145505X
    
    

}
,
}

@article{bayraktar_StabilitySampleComplexity_2025,
  title={Stability and sample complexity of divergence regularized optimal transport},
  author={Bayraktar, Erhan and Eckstein, Stephan and Zhang, Xin},
  journal={Bernoulli},
  volume={31},
  number={1},
  pages={213--239},
  year={2025},
  publisher={Bernoulli Society for Mathematical Statistics and Probability}
}

@article{barycenter_NPhard2022,
author = {Altschuler, Jason M. and Boix-Adser\`{a}, Enric},
title = {Wasserstein Barycenters Are {NP}-Hard to Compute},
journal = {SIAM Journal on Mathematics of Data Science},
volume = {4},
number = {1},
pages = {179-203},
year = {2022},

eprint = { 
    
        https://doi.org/10.1137/21M1390062
    
    

}
,
}

@book {Nutz_introEOT2022,
    AUTHOR = {Nutz, Marcel},
     TITLE = {Introduction to Entropic Optimal Transport},
 PUBLISHER = {},
    VOLUME = {129},
      YEAR = {2022},
}

@article {Christian2012,
    AUTHOR = {L\'eonard, Christian},
     TITLE = {From the {S}chr\"odinger problem to the {M}onge-{K}antorovich
              problem},
   JOURNAL = {J. Funct. Anal.},
  FJOURNAL = {Journal of Functional Analysis},
    VOLUME = {262},
      YEAR = {2012},
    NUMBER = {4},
     PAGES = {1879--1920},
MRREVIEWER = {Luca\ Granieri},
}

@article {CarlierDuvalPeyreSchmitzer2017,
    AUTHOR = {Carlier, Guillaume and Duval, Vincent and Peyr\'e, Gabriel and
              Schmitzer, Bernhard},
     TITLE = {Convergence of entropic schemes for optimal transport and
              gradient flows},
   JOURNAL = {SIAM J. Math. Anal.},
  FJOURNAL = {SIAM Journal on Mathematical Analysis},
    VOLUME = {49},
      YEAR = {2017},
    NUMBER = {2},
     PAGES = {1385--1418},
MRREVIEWER = {Kathrin\ Welker},
}

@article {BerntonGhosalNutz2022,
    AUTHOR = {Bernton, Espen and Ghosal, Promit and Nutz, Marcel},
     TITLE = {Entropic optimal transport: geometry and large deviations},
   JOURNAL = {Duke Math. J.},
  FJOURNAL = {Duke Mathematical Journal},
    VOLUME = {171},
      YEAR = {2022},
    NUMBER = {16},
     PAGES = {3363--3400},
MRREVIEWER = {Marc\ Sedjro},
}

@article {NutzWiesel2022,
    AUTHOR = {Nutz, Marcel and Wiesel, Johannes},
     TITLE = {Entropic optimal transport: convergence of potentials},
   JOURNAL = {Probab. Theory Related Fields},
  FJOURNAL = {Probability Theory and Related Fields},
    VOLUME = {184},
      YEAR = {2022},
    NUMBER = {1-2},
     PAGES = {401--424},
MRREVIEWER = {Yuxin\ Ge},
}

@article{GHOSAL2022109622,
title = {Stability of entropic optimal transport and Schrödinger bridges},
journal = {Journal of Functional Analysis},
volume = {283},
number = {9},
pages = {109622},
year = {2022},
author = {Promit Ghosal and Marcel Nutz and Espen Bernton},
}

@article{Sommerfeld_Munk_inference_OT2017,
    author = {Sommerfeld, Max and Munk, Axel},
    title = {Inference for Empirical Wasserstein Distances on Finite Spaces},
    journal = {Journal of the Royal Statistical Society Series B: Statistical Methodology},
    volume = {80},
    number = {1},
    pages = {219-238},
    year = {2017},
    month = {05},
    eprint = {https://academic.oup.com/jrsssb/article-pdf/80/1/219/49271410/jrsssb_80_1_219.pdf},
}

@inproceedings{liu2023asymptotic,
  title={Asymptotic confidence sets for random linear programs},
  author={Liu, Shuyu and Bunea, Florentina and Niles-Weed, Jonathan},
  booktitle={The Thirty Sixth Annual Conference on Learning Theory},
  pages={3919--3940},
  year={2023},
  organization={PMLR}
}

@article {limit_random_lp22,
    AUTHOR = {Klatt, Marcel and Munk, Axel and Zemel, Yoav},
     TITLE = {Limit laws for empirical optimal solutions in random linear
              programs},
   JOURNAL = {Ann. Oper. Res.},
  FJOURNAL = {Annals of Operations Research},
    VOLUME = {315},
      YEAR = {2022},
    NUMBER = {1},
     PAGES = {251--278},
MRREVIEWER = {I.\ M.\ Stancu-Minasian},
}

@book {boyd2004convex,
    AUTHOR = {Boyd, Stephen and Vandenberghe, Lieven},
     TITLE = {Convex optimization},
 PUBLISHER = {Cambridge University Press, Cambridge},
      YEAR = {2004},
     PAGES = {xiv+716},
      ISBN = {0-521-83378-7},
MRREVIEWER = {Dan\ Butnariu},
}

@article{Deville1992,
  title={Calibration estimators in survey sampling},
  author={Deville, Jean-Claude and S{\"a}rndal, Carl-Erik},
  journal={Journal of the American statistical Association},
  volume={87},
  number={418},
  pages={376--382},
  year={1992},
  publisher={Taylor \& Francis}
}

@article{Kim2010,
  title={Calibration estimation in survey sampling},
  author={Kim, Jae Kwang and Park, Mingue},
  journal={International Statistical Review},
  volume={78},
  number={1},
  pages={21--39},
  year={2010},
  publisher={Wiley Online Library}
}

@article{Wu2003,
  title={Optimal calibration estimators in survey sampling},
  author={Wu, Changbao},
  journal={Biometrika},
  volume={90},
  number={4},
  pages={937--951},
  year={2003},
  publisher={Oxford University Press}
}

@article{Montanari2005,
  title={Nonparametric model calibration estimation in survey sampling},
  author={Montanari, Giorgio E and Ranalli, M Giovanna},
  journal={Journal of the American Statistical Association},
  volume={100},
  number={472},
  pages={1429--1442},
  year={2005},
  publisher={Taylor \& Francis}
}

@article{Imbens1994,
  title={Combining micro and macro data in microeconometric models},
  author={Imbens, Guido W and Lancaster, Tony},
  journal={The Review of Economic Studies},
  volume={61},
  number={4},
  pages={655--680},
  year={1994},
  publisher={Wiley-Blackwell}
}

@article{Hellerstein1999,
  title={Imposing moment restrictions from auxiliary data by weighting},
  author={Hellerstein, Judith K and Imbens, Guido W},
  journal={Review of Economics and Statistics},
  volume={81},
  number={1},
  pages={1--14},
  year={1999},
  publisher={MIT Press 238 Main St., Suite 500, Cambridge, MA 02142-1046, USA journals}
}

@article{Sadinle2019sequentially,
  title={Sequentially additive nonignorable missing data modelling using auxiliary marginal information},
  author={Sadinle, Mauricio and Reiter, Jerome P},
  journal={Biometrika},
  volume={106},
  number={4},
  pages={889--911},
  year={2019},
  publisher={Oxford University Press}
}

@article{Akande2021leveraging,
  title={Leveraging auxiliary information on marginal distributions in nonignorable models for item and unit nonresponse},
  author={Akande, Olanrewaju and Madson, Gabriel and Hillygus, D Sunshine and Reiter, Jerome P},
  journal={Journal of the Royal Statistical Society Series A: Statistics in Society},
  volume={184},
  number={2},
  pages={643--662},
  year={2021},
  publisher={Oxford University Press}
}

@article{Nevo2003,
  title={Using weights to adjust for sample selection when auxiliary information is available},
  author={Nevo, Aviv},
  journal={Journal of Business \& Economic Statistics},
  volume={21},
  number={1},
  pages={43--52},
  year={2003},
  publisher={Taylor \& Francis}
}

@article{Chen2008,
 author = {Xiaohong Chen and Han Hong and Alessandro Tarozzi},
 journal = {The Annals of Statistics},
 number = {2},
 pages = {808--843},
 publisher = {Institute of Mathematical Statistics},
 title = {Semiparametric Efficiency in GMM Models with Auxiliary Data}, 
 volume = {36},
 year = {2008}
}

@article{Qian2025,
  title={A deep generative framework for joint households and individuals population synthesis},
  author={Qian, Xiao and Gangwal, Utkarsh and Dong, Shangjia and Davidson, Rachel},
  journal={Applied Soft Computing},
  pages={114375},
  year={2025},
  publisher={Elsevier}
}

@article{Sane2025,
  title={Population synthesis with deep generative model: {A} joint household-individual approach},
  author={San{\'e}, Abdoul Razac and Belaroussi, Rachid and Hankach, Pierre and Vandanjon, Pierre-Olivier},
  journal={Computational Urban Science},
  volume={5},
  number={1},
  pages={34},
  year={2025},
  publisher={Springer}
}

@article{GMM1982,
 author = {Lars Peter Hansen},
 journal = {Econometrica},
 number = {4},
 pages = {1029--1054},
 publisher = {[Wiley, Econometric Society]},
 title = {Large Sample Properties of Generalized Method of Moments Estimators}, 
 volume = {50},
 year = {1982}
}

@article{Wu2001model,
  title={A model-calibration approach to using complete auxiliary information from survey data},
  author={Wu, Changbao and Sitter, Randy R},
  journal={Journal of the American Statistical Association},
  volume={96},
  number={453},
  pages={185--193},
  year={2001},
  publisher={Taylor \& Francis}
}

@article {Benamou_brener2000,
    AUTHOR = {Benamou, Jean-David and Brenier, Yann},
     TITLE = {A computational fluid mechanics solution to the
              {M}onge-{K}antorovich mass transfer problem},
   JOURNAL = {Numer. Math.},
  FJOURNAL = {Numerische Mathematik},
    VOLUME = {84},
      YEAR = {2000},
    NUMBER = {3},
     PAGES = {375--393},
MRREVIEWER = {Enrique\ Fern\'andez Cara},
}

@book{mccann1994convexity,
  title={A convexity theory for interacting gases and equilibrium crystals},
  author={McCann, Robert John},
  year={1994},
  publisher={Princeton University}
}

@article{Otto2000GeneralizationOA,
  title={Generalization of an Inequality by Talagrand and Links with the Logarithmic Sobolev Inequality},
  author={Felix Otto and C{\'e}dric Villani},
  journal={Journal of Functional Analysis},
  year={2000},
  volume={173},
  pages={361-400},
}

@article{Lott2004RicciCF,
  title={Ricci curvature for metric-measure spaces via optimal transport},
  author={John Lott and C{\'e}dric Villani},
  journal={Annals of Mathematics},
  year={2004},
  volume={169},
  pages={903-991},
}

@article {Figalli2010Invent,
    AUTHOR = {Figalli, A. and Maggi, F. and Pratelli, A.},
     TITLE = {A mass transportation approach to quantitative isoperimetric
              inequalities},
   JOURNAL = {Invent. Math.},
  FJOURNAL = {Inventiones Mathematicae},
    VOLUME = {182},
      YEAR = {2010},
    NUMBER = {1},
     PAGES = {167--211},
MRREVIEWER = {Lorenzo\ Brasco},
}

@article {Kim2010JEMS,
    AUTHOR = {Kim, Young-Heon and McCann, Robert J.},
     TITLE = {Continuity, curvature, and the general covariance of optimal
              transportation},
   JOURNAL = {J. Eur. Math. Soc. (JEMS)},
  FJOURNAL = {Journal of the European Mathematical Society (JEMS)},
    VOLUME = {12},
      YEAR = {2010},
    NUMBER = {4},
     PAGES = {1009--1040},
MRREVIEWER = {Lorenzo\ Brasco},
}

@article {Figall2011duke,
    AUTHOR = {Carrillo, J. A. and DiFrancesco, M. and Figalli, A. and
              Laurent, T. and Slep\v cev, D.},
     TITLE = {Global-in-time weak measure solutions and finite-time
              aggregation for nonlocal interaction equations},
   JOURNAL = {Duke Math. J.},
  FJOURNAL = {Duke Mathematical Journal},
    VOLUME = {156},
      YEAR = {2011},
    NUMBER = {2},
     PAGES = {229--271},
MRREVIEWER = {Teemu\ Lukkari},
}

@article {Caffarelli2013Reine,
    AUTHOR = {Caffarelli, Luis and Figalli, Alessio},
     TITLE = {Regularity of solutions to the parabolic fractional obstacle
              problem},
   JOURNAL = {J. Reine Angew. Math.},
  FJOURNAL = {Journal f\"ur die Reine und Angewandte Mathematik. [Crelle's
              Journal]},
    VOLUME = {680},
      YEAR = {2013},
     PAGES = {191--233},
MRREVIEWER = {Luca\ Lorenzi},
}

@article {Figalli2015CMP,
    AUTHOR = {Figalli, A. and Fusco, N. and Maggi, F. and Millot, V. and
              Morini, M.},
     TITLE = {Isoperimetry and stability properties of balls with respect to
              nonlocal energies},
   JOURNAL = {Comm. Math. Phys.},
  FJOURNAL = {Communications in Mathematical Physics},
    VOLUME = {336},
      YEAR = {2015},
    NUMBER = {1},
     PAGES = {441--507},
MRREVIEWER = {Isabel\ M. C. Salavessa},
}

@article {MR3314838,
    AUTHOR = {Colombo, Maria and De Pascale, Luigi and Di Marino, Simone},
     TITLE = {Multimarginal optimal transport maps for one-dimensional
              repulsive costs},
   JOURNAL = {Canad. J. Math.},
  FJOURNAL = {Canadian Journal of Mathematics. Journal Canadien de
              Math\'{e}matiques},
    VOLUME = {67},
      YEAR = {2015},
    NUMBER = {2},
     PAGES = {350--368},
MRREVIEWER = {Lorenzo Brasco},
}

@article {MR2110613,
    AUTHOR = {Ekeland, Ivar},
     TITLE = {An optimal matching problem},
   JOURNAL = {ESAIM Control Optim. Calc. Var.},
  FJOURNAL = {ESAIM. Control, Optimisation and Calculus of Variations},
    VOLUME = {11},
      YEAR = {2005},
    NUMBER = {1},
     PAGES = {57--71},
MRREVIEWER = {Luigi De Pascale},
}

@article {MR2564439,
    AUTHOR = {Chiappori, Pierre-Andr\'{e} and McCann, Robert J. and Nesheim,
              Lars P.},
     TITLE = {Hedonic price equilibria, stable matching, and optimal
              transport: equivalence, topology, and uniqueness},
   JOURNAL = {Econom. Theory},
  FJOURNAL = {Economic Theory},
    VOLUME = {42},
      YEAR = {2010},
    NUMBER = {2},
     PAGES = {317--354},
MRREVIEWER = {Wilfrid Gangbo},
}

@article{Carlier2010MatchingFT,
  title={Matching for teams},
  author={Guillaume Carlier and Ivar Ekeland},
  journal={Economic Theory},
  year={2010},
  volume={42},
  pages={397-418}
}

@article{de2021diffusion,
  title={Diffusion schr{\"o}dinger bridge with applications to score-based generative modeling},
  author={De Bortoli, Valentin and Thornton, James and Heng, Jeremy and Doucet, Arnaud},
  journal={Advances in Neural Information Processing Systems},
  volume={34},
  pages={17695--17709},
  year={2021}
}

@inproceedings{wang2021deep,
  title={Deep generative learning via Schr{\"o}dinger bridge},
  author={Wang, Gefei and Jiao, Yuling and Xu, Qian and Wang, Yang and Yang, Can},
  booktitle={International conference on machine learning},
  pages={10794--10804},
  year={2021},
  organization={PMLR}
}

@article{hamdouche2023generative,
  title={Generative modeling for time series via Schr{\"o}dinger bridge},
  author={Hamdouche, Mohamed and Henry-Labordere, Pierre and Pham, Huy{\^e}n},
  journal={arXiv preprint arXiv:2304.05093},
  year={2023}
}

@article{lambert2022variational,
  title={Variational inference via Wasserstein gradient flows},
  author={Lambert, Marc and Chewi, Sinho and Bach, Francis and Bonnabel, Silv{\`e}re and Rigollet, Philippe},
  journal={Advances in Neural Information Processing Systems},
  volume={35},
  pages={14434--14447},
  year={2022}
}

@inproceedings{diao2023forward,
  title={Forward-backward Gaussian variational inference via JKO in the Bures-Wasserstein space},
  author={Diao, Michael Ziyang and Balasubramanian, Krishna and Chewi, Sinho and Salim, Adil},
  booktitle={International Conference on Machine Learning},
  pages={7960--7991},
  year={2023},
  organization={PMLR}
}

@inproceedings{
pydi2021the,
title={The Many Faces of Adversarial Risk},
author={Muni Sreenivas Pydi and Varun Jog},
booktitle={Thirty-Fifth Conference on Neural Information Processing Systems},
year={2021},
}

@article{Trillos2020AdversarialCN,
  author  = {Nicol{\'a}s {Garc{\'i}a Trillos} and Ryan W. Murray},
  title   = {Adversarial Classification: Necessary Conditions and Geometric Flows},
  journal = {Journal of Machine Learning Research},
  year    = {2022},
  volume  = {23},
  number  = {187},
  pages   = {1--38},
}

@article{jakwang2023JMLR,
author = {Trillos, Nicol\'{a}s Garc\'{\i}a and Kim, Jakwang and Jacobs, Matt},
title = {The multimarginal optimal transport formulation of adversarial multiclass classification},
year = {2024},
issue_date = {January 2023},
publisher = {JMLR.org},
volume = {24},
number = {1},
journal = {J. Mach. Learn. Res.},
month = {mar},
articleno = {45},
numpages = {56},
}

@article{jakwang_2024existence,
  title={On the existence of solutions to adversarial training in multiclass classification},
  author={Garc\'{\i}a Trillos, Nicol\'{a}s and Jacobs, Matt and Kim, Jakwang},
  journal={European Journal of Applied Mathematics},
  pages={1--21},
  year={2024},
  publisher={Cambridge University Press}
}

@article{jakwang_2024optimal,
author = {Garc\'{\i}a Trillos, Nicol\'{a}s and Jacobs, Matt and Kim, Jakwang and Werenski, Matthew},
title = {An optimal transport approach for computing adversarial training lower bounds in multiclass classification},
year = {2024},
issue_date = {January 2024},
publisher = {JMLR.org},
volume = {25},
number = {1},
journal = {J. Mach. Learn. Res.},
month = jan,
articleno = {393},
numpages = {45},
}

@article{mei2018mean,
  title={A mean field view of the landscape of two-layer neural networks},
  author={Mei, Song and Montanari, Andrea and Nguyen, Phan-Minh},
  journal={Proceedings of the National Academy of Sciences},
  volume={115},
  number={33},
  pages={E7665--E7671},
  year={2018},
  publisher={National Acad Sciences}
}

@article{chizat2018global,
  title={On the global convergence of gradient descent for over-parameterized models using optimal transport},
  author={Chizat, Lenaic and Bach, Francis},
  journal={Advances in neural information processing systems},
  volume={31},
  year={2018}
}

@article{sirignano2020mean,
  title={Mean field analysis of neural networks: A law of large numbers},
  author={Sirignano, Justin and Spiliopoulos, Konstantinos},
  journal={SIAM Journal on Applied Mathematics},
  volume={80},
  number={2},
  pages={725--752},
  year={2020},
  publisher={SIAM}
}

@article {Rotskoff2022CPAM,
    AUTHOR = {Rotskoff, G. M. and Vanden-Eijnden, E.},
     TITLE = {Trainability and accuracy of artificial neural networks: an
              interacting particle system approach},
   JOURNAL = {Comm. Pure Appl. Math.},
  FJOURNAL = {Communications on Pure and Applied Mathematics},
    VOLUME = {75},
      YEAR = {2022},
    NUMBER = {9},
     PAGES = {1889--1935},
}

@article{Gangbo1998OptimalMF,
  title={Optimal maps for the multidimensional Monge-Kantorovich problem},
  author={Wilfrid Gangbo and Andrzej \'{S}wi{e}ch},
  journal={Communications on Pure and Applied Mathematics},
  year={1998},
  volume={51},
  pages={23-45}
}

@article{Kim2013MultimarginalOT,
  title={Multi-marginal optimal transport on Riemannian manifolds},
  author={Young-Heon Kim and Brendan Pass},
  journal={American Journal of Mathematics},
  year={2013},
  volume={137},
  pages={1045 - 1060}
}

@article{KimPass-SIAM2014,
author = {Kim, Young-Heon and Pass, Brendan},
title = {A General Condition for Monge Solutions in the Multi-Marginal Optimal Transport Problem},
journal = {SIAM Journal on Mathematical Analysis},
volume = {46},
number = {2},
pages = {1538-1550},
year = {2014},

eprint = { 
    
        https://doi.org/10.1137/130930443
    
    

}
,
}

@article {MR3423275,
    AUTHOR = {Pass, Brendan},
     TITLE = {Multi-marginal optimal transport: theory and applications},
   JOURNAL = {ESAIM Math. Model. Numer. Anal.},
  FJOURNAL = {ESAIM. Mathematical Modelling and Numerical Analysis},
    VOLUME = {49},
      YEAR = {2015},
    NUMBER = {6},
     PAGES = {1771--1790},
}

@article {MR3482267,
    AUTHOR = {Kitagawa, Jun and Pass, Brendan},
     TITLE = {The multi-marginal optimal partial transport problem},
   JOURNAL = {Forum Math. Sigma},
  FJOURNAL = {Forum of Mathematics. Sigma},
    VOLUME = {3},
      YEAR = {2015},
     PAGES = {Paper No. e17, 28},
MRREVIEWER = {Edgard Pimentel},
}

@article{seidl2007strictly,
  title={Strictly correlated electrons in density-functional theory: A general formulation with applications to spherical densities},
  author={Seidl, Michael and Gori-Giorgi, Paola and Savin, Andreas},
  journal={Physical Review A},
  volume={75},
  number={4},
  pages={042511},
  year={2007},
  publisher={APS}
}

@article{buttazzo2012optimal,
  title={Optimal-transport formulation of electronic density-functional theory},
  author={Buttazzo, Giuseppe and De Pascale, Luigi and Gori-Giorgi, Paola},
  journal={Physical Review A},
  volume={85},
  number={6},
  pages={062502},
  year={2012},
  publisher={APS}
}

@article {MR3020313,
    AUTHOR = {Cotar, Codina and Friesecke, Gero and Kl\"{u}ppelberg, Claudia},
     TITLE = {Density functional theory and optimal transportation with
              {C}oulomb cost},
   JOURNAL = {Comm. Pure Appl. Math.},
  FJOURNAL = {Communications on Pure and Applied Mathematics},
    VOLUME = {66},
      YEAR = {2013},
    NUMBER = {4},
     PAGES = {548--599},
MRREVIEWER = {Gabriel Stoltz},
}

@article{mendl2013kantorovich,
  title={Kantorovich dual solution for strictly correlated electrons in atoms and molecules},
  author={Mendl, Christian B and Lin, Lin},
  journal={Physical Review B},
  volume={87},
  number={12},
  pages={125106},
  year={2013},
  publisher={APS}
}

@article {MR2801182,
    AUTHOR = {Agueh, Martial and Carlier, Guillaume},
     TITLE = {Barycenters in the {W}asserstein space},
   JOURNAL = {SIAM J. Math. Anal.},
  FJOURNAL = {SIAM Journal on Mathematical Analysis},
    VOLUME = {43},
      YEAR = {2011},
    NUMBER = {2},
     PAGES = {904--924},
MRREVIEWER = {Beno\^{\i}t Kloeckner},
}

@inproceedings{cuturi2014fast,
  title={Fast computation of Wasserstein barycenters},
  author={Cuturi, Marco and Doucet, Arnaud},
  booktitle={International conference on machine learning},
  pages={685--693},
  year={2014},
  organization={PMLR}
}

@article{benamou2015iterative,
  title={Iterative Bregman projections for regularized transportation problems},
  author={Benamou, Jean-David and Carlier, Guillaume and Cuturi, Marco and Nenna, Luca and Peyr{\'e}, Gabriel},
  journal={SIAM Journal on Scientific Computing},
  volume={37},
  number={2},
  pages={A1111--A1138},
  year={2015},
  publisher={SIAM}
}

@article {MR3423268,
    AUTHOR = {Carlier, Guillaume and Oberman, Adam and Oudet, Edouard},
     TITLE = {Numerical methods for matching for teams and {W}asserstein
              barycenters},
   JOURNAL = {ESAIM Math. Model. Numer. Anal.},
  FJOURNAL = {ESAIM. Mathematical Modelling and Numerical Analysis},
    VOLUME = {49},
      YEAR = {2015},
    NUMBER = {6},
     PAGES = {1621--1642},
MRREVIEWER = {Alessio Brancolini},
}

@article{srivastava2018scalable,
  title={Scalable Bayes via barycenter in Wasserstein space},
  author={Srivastava, Sanvesh and Li, Cheng and Dunson, David B},
  journal={The Journal of Machine Learning Research},
  volume={19},
  number={1},
  pages={312--346},
  year={2018},
  publisher={JMLR. org}
}

@article {MR4809472,
    AUTHOR = {Pass, Brendan and Vargas-Jim\'enez, Adolfo},
     TITLE = {A general framework for multi-marginal optimal transport},
   JOURNAL = {Math. Program.},
  FJOURNAL = {Mathematical Programming},
    VOLUME = {208},
      YEAR = {2024},
    NUMBER = {1-2},
     PAGES = {75--110},
}

@misc{pass2025dynamicalformulationmultimarginaloptimal,
      title={A dynamical formulation of multi-marginal optimal transport}, 
      author={Brendan Pass and Yair Shenfeld},
      year={2025},
      eprint={2509.22494},
      archivePrefix={arXiv},
      primaryClass={math.OC},
}

@book{politis1999subsampling,
  title={Subsampling},
  author={Politis, D.N. and Romano, J.P. and Wolf, M.},
  isbn={9780387988542},
  lccn={99015017},
  series={Springer Series in Statistics},
  year={1999},
  publisher={Springer New York}
}

@article{politis1994large,
  title={Large sample confidence regions based on subsamples under minimal assumptions},
  author={Politis, Dimitris N and Romano, Joseph P},
  journal={The Annals of Statistics},
  pages={2031--2050},
  year={1994},
  publisher={JSTOR}
}

@book{efron1994introduction,
  title={An Introduction to the Bootstrap},
  author={Efron, Bradley and Tibshirani, Robert J},
  year={1994},
  publisher={Chapman and Hall/CRC}
}

@article{shao1994bootstrap,
 author = {Jun Shao},
 journal = {Proceedings of the American Mathematical Society},
 number = {4},
 pages = {1251--1262},
 publisher = {American Mathematical Society},
 title = {Bootstrap Sample Size in Nonregular Cases},
 volume = {122},
 year = {1994}
}

@article{andrews2000inconsistency,
  title={Inconsistency of the bootstrap when a parameter is on the boundary of the parameter space},
  author={Andrews, Donald WK},
  journal={Econometrica},
  pages={399--405},
  year={2000},
  publisher={JSTOR}
}

@article{fang2019inference,
  title={Inference on directionally differentiable functions},
  author={Fang, Zheng and Santos, Andres},
  journal={The Review of Economic Studies},
  volume={86},
  number={1},
  pages={377--412},
  year={2019},
  publisher={Oxford University Press}
}

@article{moonboot,
    title = {moonboot: An R Package Implementing m-out-of-n Bootstrap Methods},
    author = {Christoph Dalitz and Felix Lögler},
    year = {2025},
    journal = {The R Journal},
    volume = {17},
    issue = {3},
    pages = {125-137},
    note = {https://github.com/cdalitz/moonboot/},
  }

@article{Dalitz_2025,
   title={moonboot: An R Package Implementing m-out-of-n Bootstrap Methods},
   volume={17},
   number={3},
   journal={The R Journal},
   publisher={The R Foundation},
   author={Dalitz, Christoph and Lögler, Felix},
   year={2025},
   month=oct, pages={125–137} }

@article{Yang2007CV,
 author = {Yuhong Yang},
 journal = {The Annals of Statistics},
 number = {6},
 pages = {2450--2473},
 publisher = {Institute of Mathematical Statistics},
 title = {Consistency of Cross Validation for Comparing Regression Procedures},
 volume = {35},
 year = {2007}
}

@article{vdVaart2006,
title = {Oracle inequalities for multi-fold cross validation},
title = {},
author = {Aad W. {van der Vaart} and Sandrine {Dudoit} and Mark J. {van der Laan}},
pages = {351--371},
volume = {24},
number = {3},
journal = {Statistics \& Decisions},
year = {2006},
lastchecked = {2026-03-22}
}

\appendices

\section{Appendix: Optimal transport projection framework}\label{sec:appendix:deferred proofs}

The problem \eqref{eq: Wasserstein projection} can be generalized as follows. Let $\mu^i \in \mathcal{P}(\mathcal{X}_i)$ for $i=1, \dots, K$ and $\pi^0 \in \mathcal{P}(\mathcal{Z})$. Assume that $c : \mathcal{X}_1 \times \dots \times \mathcal{X}_K \times \mathcal{Z} \to \mathbb{R}$ is continuous. The problem of interest is 
\begin{equation}\label{eq: primal}
    \min_{\pi \in \Pi(\bm{\mu})} \min_{\gamma \in \Pi(\pi, \pi^0)} \int c(\boldsymbol{x}, \boldsymbol{z}) d\gamma(\boldsymbol{x}, \boldsymbol{z}).
\end{equation}
This is the optimal transport projection of $\pi^0$ onto $\Pi(\bm{\mu})$. This projection problem can be equivalently formulated as the following MOT problem:
\begin{equation}\label{eq: primal MOT version}
    \min_{\gamma \in \Pi(\bm{\mu}, \pi^0)} \int  c d\gamma
\end{equation}
where the marginal constraints are given as
\[
    (\mathrm{P}_{x_i})_\#(\gamma) = \mu^i \text{ for $i=1, \dots, K$}, \quad (\mathrm{P}_{\bm{z}})_\#(\gamma) = \pi^0.
\]
Under the assumption on $c$ as above, \eqref{eq: primal MOT version} has a solution \cite[Theorem 4.1]{villani2008optimal}. Let us denote an optimal multimarginal coupling by $\gamma^*$. Then, a solution for \eqref{eq: primal}, denoted by $\pi^*$, is attained from $\gamma^*$ by
\begin{equation*}\label{eq: pi^*}
    \pi^* := (\mathrm{P}_{\boldsymbol{x}})_\# (\gamma^*)
\end{equation*}
where $\mathrm{P}_{\boldsymbol{x}}$ is the projection onto $\mathcal{X}_1 \times \dots \times \mathcal{X}_K$.

The first theorem concerns the stability of \eqref{eq: primal MOT version}, hence of \eqref{eq: primal} in terms of input measures. The convergence result is already well known in the literature: see \cite[Theorem 5.20]{villani2008optimal}.

\begin{theorem}\label{theorem: stability of MOT}
Let $m=m(t), n=n(t) \to \infty$ as $t \to \infty$. Assume that $\pi^0_m \to \pi^0, \mu^i_n \to \mu^i$ for $1\leq i \leq K$ weakly and
\[
    \liminf_{t \to \infty} \min_{\gamma \in \Pi(\mu_n^1, \dots, \mu_n^K, \pi_m^0)} \int  c d\gamma < \infty.
\]
If $\{\gamma^*_t\}$ is the sequence of minimizers for \eqref{eq: primal MOT version} with $\mu^1_{n(t)}, \dots, \mu^K_{n(t)}, \pi^0_{m(t)}$, then there is a subsequence that converges to $\gamma^*$, and $\gamma^*$ is a minimizer for \eqref{eq: primal MOT version} with $\mu^1, \dots, \mu^K, \pi^0$. Furthermore, there is a subsequence of $\{\pi^*_t := (\mathrm{P}_{\boldsymbol{x}})_\# (\gamma^*_t)\}$ that converges to some $\pi^*$, and $\pi^*$ is a solution for \eqref{eq: primal} with $\mu^1, \dots, \mu^K, \pi^0$.
\end{theorem}

For separable cost functions which are written as the sum of cost functions of pairwise coordinates, one can obtain $\gamma^*$ and $\pi^*$ more precisely by shadow proposed by \citet{Eckstein_Nutz_2022}. Note that there are possibly other optimal solutions not obtained by shadow.

\begin{theorem}[Separable case]\label{thm: separable case}
Assume that $\mathcal{Z} = \mathcal{Z}_1 \times \dots \times \mathcal{Z}_K$ and a cost function is separable in the sense that
\[
    c(\boldsymbol{x}, \boldsymbol{z}) = \sum_{i=1}^K c_i(x_i, z_i)
\]
for $c_i : \mathcal{X}_i \times \mathcal{Z}_i \to \mathbb{R}$. For each $i=1, \dots, K$, let $\nu^i := (\mathrm{P}_{z_i})_\# (\pi^0)$ where $\mathrm{P}_{z_i}$ is the projection from $\mathcal{Z}$ to $\mathcal{Z}_i$. Then
\begin{equation}\label{eq: separable primal}
    \min_{\gamma \in \Pi(\bm{\mu}, \pi^0)} \int  c d\gamma= \sum_{i=1}^K \min_{\gamma^i \in \Pi(\mu^i, \nu^i)} \int c_i d\gamma^i.
\end{equation}

Let $\gamma^{i,*}$ be an optimal coupling for $\min_{\gamma^i \in \Pi(\mu^i, \nu^i)} \int c_i d\gamma^i$, which can be written as $\gamma^{i,*}(dx_i, dz_i) = \nu^i(dz_i) \kappa^i(dx_i |z_i)$ by disintegration formula. Then
\begin{equation*}
    \gamma^*(d \boldsymbol{x}, d\boldsymbol{z}) = \pi^0(d\boldsymbol{z}) \kappa(d\boldsymbol{x} | \boldsymbol{z})
\end{equation*}
is optimal where 
\[
    \kappa(d\boldsymbol{x} | \boldsymbol{z}):=  \kappa_1(dx_1 |z_1) \otimes \dots \otimes \kappa_K(dx_K |z_K).
\]
In particular, an optimal $\pi^*$ for \eqref{eq: primal} is also obtained by
    \[
        \pi^*(d\boldsymbol{x}) = \int_{\mathcal{Z}} \pi^0(d\boldsymbol{z}) \kappa(d\boldsymbol{x} | \boldsymbol{z}).
    \]
\end{theorem}

Recall the problem setting. Let $\pi^0 \in \Pi(\bm{\mu})$ be the true coupling generating the data. We denote by $\left\{ \bm{Z}_j=(Z_{1j}, \dots, Z_{Kj}) : j=1, \dots, m \right\}$ and $\{ X_{1j} : j=1, \dots, n \}, \dots, \{ X_{Kj} : j=1, \dots, n \}$ the set of coupled data i.i.d. from $\pi^0$, and the sets of marginal data, respectively. We use
\[
    \pi^0_m:= \sum_{j=1}^m \delta_{\bm{Z}_j}, \quad \mu^i_n:= \sum_{j=1}^n \delta_{X_{ij}} \text{ for $i=1, \dots, K$.}
\]
The problem is to reconstruct a coupling which should be consistent with the marginal observations. In other words, it should lie in $\Pi(\bm{\mu}_n)$, and close to $\pi^0_m$ as much as possible.

Using the optimal transport-projection framework, or the equivalent multimarginal optimal transport formulation, one can construct a coupling over $\mu^i_n$'s as follows: first, compute $\hat{\gamma}$, which is a solution for
\[
    \min_{\gamma \in \Pi(\bm{\mu}_n, \pi^0_m)} \int  c d\gamma,
\]
and the estimated coupling $\hat{\pi}$ is obtained by the push forward measure of $\hat{\gamma}$ by $\mathrm{P}_{\boldsymbol{x}}$ as
\[
    \hat{\pi} = (\mathrm{P}_{\boldsymbol{x}})_\# (\hat{\gamma}).
\]

As a corollary of \Cref{theorem: stability of MOT}, we can derive the consistency of an estimator $\hat{\pi}$.

\begin{corollary}[Consistency]
Assume that $\mathcal{Z} = \mathcal{X}_1 \times \dots \times \mathcal{X}_K$, $c \in L^1((\otimes \bm{\mu}_n) \otimes \pi^0)$, $c$ is nonnegative and $c(\bm{x}, \bm{x})=0$ for all $\bm{x} \in \mathcal{Z}$, and $\pi^0 \in \Pi(\bm{\mu})$. Let $m=m(t), n=n(t) \to \infty$ as $t \to \infty$, and $\hat{\pi}_t$ be a solution for \eqref{eq: primal} with empirical input. Then there is a subsequence of $\{ \hat{\pi}_t \}$ that converges to $\pi^0$.
\end{corollary}

\subsection{Entropic optimal transport}\label{sec:appendix: EOT}
The entropic regularization of \eqref{eq: primal} with a general cost $c$ is defined as
\begin{equation}\label{eq: entropic primal}
    \min_{\pi \in \Pi(\bm{\mu})} \min_{\gamma \in \Pi(\pi, \pi^0)} \int c(\boldsymbol{x}, \boldsymbol{z}) d\gamma(\boldsymbol{x}, \boldsymbol{z}) + \eta \mathrm{KL} \left( \gamma \, \Vert\,  (\otimes \bm{\mu}_n) \otimes \pi^0 \right).
\end{equation}
It is equivalent to 
\begin{equation}\label{eq: entropic MOT}
    \min_{\gamma \in \Pi(\bm{\mu}, \pi^0)} \int c d \gamma  + \eta \mathrm{KL} \left( \gamma \, \Vert\,  (\otimes \bm{\mu}_n) \otimes \pi^0 \right), 
\end{equation}
which is also regarded as the entropic regularization of \eqref{eq: primal MOT version}. 
Due to the strong convexity of the KL divergence, there is a unique optimal coupling $\gamma^{\eta}_{\star}$ for \eqref{eq: entropic MOT}. This leads to a unique solution for \eqref{eq: entropic primal}, denoted by $\pi^{\eta}_{\star}$. Note that since $(\mathrm{P}_{\bm{x}})_\# \gamma^{\eta}_{\star} \ll (\otimes \bm{\mu})$, $\pi^{\eta}_{\star}$ is absolutely continuous with respect to $(\otimes \bm{\mu})$.

For the separable cost case, i.e., $c(\bm{x}, \bm{z}) = \sum c_i(x_i, z_i)$, one can utilize the entropic shadow approach. Given a coupling $\gamma \in \Pi(\bm{\mu}, \pi^0)$, let $\gamma^i(dx_i, dz_i) := (\mathrm{P}_{x_i}, \mathrm{P}_{z_i})_\# (\gamma)$ be its marginal coupling between $\mu^i$ and $\nu^i = (\mathrm{P}_{z_i})_\# (\pi^0)$. Consider
\begin{align}\label{eq: local entropic MOT}
    \min_{\gamma \in \Pi(\bm{\mu}, \pi^0)} \sum_{i=1}^K \int  c_i(x_i, z_i)  \gamma(d \boldsymbol{x}, d \boldsymbol{z)} + \eta \mathrm{KL}(\gamma^i \, \Vert\, \mu^i \otimes \nu^i).
\end{align}
For $i=1, \dots, K$, let $\gamma^{i,\eta} \in \Pi(\mu^i, \nu^i)$ be optimal for 
\[
    \min_{\gamma^{i} \in \Pi(\mu^i, \nu^i)} \int  c_i  d\gamma^i  + \eta \mathrm{KL}(\gamma^i \, \Vert\, \mu^i \otimes \nu^i),
\]
and $\kappa^{i,\eta}(dx_i |z_i) \nu^i(dz_i) = \gamma^{i, \eta}(dx_i, d z_i)$ by disintegration formula. The MOT entropic coupling for \eqref{eq: local entropic MOT}, denoted by $\gamma^\eta_{\dagger}$, is constructed by
\[
    \gamma^\eta_{\dagger} (d\boldsymbol{x}, d \boldsymbol{z} ) = \pi^0(d \boldsymbol{z}) \kappa^\eta(d \boldsymbol{x} |\boldsymbol{z})
\]
where $\kappa^\eta = \kappa^{1,\eta}(dx_1 |z_1) \otimes \dots \otimes \kappa^{K,\eta}(dx_K |z_K)$. Then, similarly, the projection solution is obtained by $\pi^\eta_{\dagger}:=(\mathrm{P}_{\bm{x}})_\# \gamma^\eta_{\dagger}$.

As mentioned, due to the entropic bias, both estimators $\pi^\eta_{\star}$ and $\gamma^{\eta}_{\dagger}$ are not necessarily the same as $\pi^*$. It is known that $\pi^\eta_{\star}$ and $\gamma^{\eta}_{\dagger}$ converge to $\pi^*$ as $\eta \to 0$ in $\Gamma$-convergence sense under mild conditions.

\begin{definition}\label{def: Gamma convergence}
Given the sequence of functionals $\{ F_n \}$ over a topological space $\mathcal{X}$, the $\Gamma$-lower limit $\Gamma \text{-}\liminf_{n \to \infty} F_n : \mathcal{X} \to \mathbb{R}\cup\{\pm \infty\}$ as
\[
    \left( \Gamma \text{-}\liminf_{n \to \infty} F_n \right) (x) := \inf \left\{ \liminf_{n \to \infty} F_n (x_n) : x_n \to x \right\},
\]
and the $\Gamma$-upper limit $\Gamma \text{-}\limsup_{n \to \infty} F_n : \mathcal{X} \to \mathbb{R}\cup\{\pm \infty\}$ as
\[
    \left( \Gamma \text{-}\limsup_{n \to \infty} F_n \right) (x) := \inf \left\{ \limsup_{n \to \infty} F_n (x_n) : x_n \to x \right\},
\]
are defined for every $x \in \mathcal{X}$. The $\Gamma$-limit is $F : \mathcal{X} \to \mathbb{R}\cup\{\pm \infty\}$ if for any $x \in \mathcal{X}$,
\[
    F(x) := \left(\Gamma \text{-}\lim_{n \to \infty} F_n \right)(x)= \left( \Gamma \text{-}\liminf_{n \to \infty} F_n \right) (x) = \left( \Gamma \text{-}\limsup_{n \to \infty} F_n \right) (x).
\]
\end{definition}

\begin{proposition}[$\Gamma$-convergence]
Assume that either $c$ is non-negative and lower semicontinuous or continuous such that $c \in L^1((\otimes \bm{\mu}_n) \otimes \pi^0)$. Then, as $\eta \to 0$, both \eqref{eq: entropic MOT} and \eqref{eq: local entropic MOT} $\Gamma$-converge to \eqref{eq: primal MOT version} and \eqref{eq: separable primal}, respectively. Furthermore, there is a subsequence of $\pi^\eta_{\star}$ which converges to a solution for \eqref{eq: primal}. Similarly, there is a subsequence of $\pi^\eta_{\dagger}$ which converges to a solution for \eqref{eq: separable primal}.
\end{proposition}

\subsection{Limit distribution for finite support case}\label{sec:appendix:limit for finite case}
Recall the linear programming of interest
\begin{equation}\label{def: general finite linear primal}
    \min_{\gamma \in \mathbb{R}^{S_*}} \langle \boldsymbol{c}, \gamma \rangle,\quad \textrm{s.t.}\ \bm{A}\gamma= \bm{b}_{n,m},\ \gamma\geq\bm{0}
\end{equation}
where $\boldsymbol{c} \in \mathbb{R}^{S_*}$ is the vectorization of $s_1 \times \dots \times s_K \times s_0$ cost tensor (defined as $d_{\mathcal{Z}, p}^p$), $\bm{A} \in  \mathbb{R}^{S_+ \times S_*}$ is the projection matrix corresponding the marginal constraints, and 
\begin{equation*}
    \bm{b}_{n,m}:=
\begin{bmatrix}
\bm{\mu}_n\\
\pi^0_m
\end{bmatrix},
\quad 
\bm{b} := 
\begin{bmatrix}
\bm{\mu}\\
\pi^0
\end{bmatrix}\in \mathbb{R}^{s_1} \times \dots \times \mathbb{R}^{s_K} \times \mathbb{R}^{s_0}
\end{equation*}
be the vector of marginal distributions.

The lemma below is a black box to attain \Cref{thm: limit distribution of projection}.

\begin{lemma}{\cite[Theorem 4]{liu2023asymptotic}}\label{lemma: aymptotic of linear programming}
Suppose \eqref{def: general finite linear primal} satisfies \Cref{assumption: finite support case}. If $\bm{b}_{n,m}$ has the distribution limit $\mathbb{G}$ and the rate $r_{n,m}$, and $\gamma(\bm{b})$ is singleton, then
\begin{equation*}
    r_{n,m} ( \gamma(\bm{b}_{n,m}) - \gamma(\bm{b}) ) \overset{d}{\longrightarrow} p^*_{\bm{b}}(\mathbb{G})
\end{equation*}
where $p^*_{\bm{b}}(\mathbb{G})$ is the set of optimal solutions to the following linear program:
\begin{equation*}
    \min_{p \in \mathbb{R}^{S_*}} \langle \boldsymbol{c}, p \rangle,\quad \textrm{s.t.}\ \bm{A}p= \mathbb{G},\ p_{ij} \geq 0 \text{ for all } (ij) \notin \text{spt}(\gamma(\bm{b})).
\end{equation*}
Here $\gamma(\bm{b}_{n,m}) - \gamma(\bm{b}) := \{ \gamma -\gamma(\bm{b}) : \gamma \in \gamma(\bm{b}_{n,m}) \}$.
\end{lemma}

The next lemma verifies that our problem satisfies \Cref{assumption: finite support case}.

\begin{lemma}\label{lemma: assumption check}
\eqref{def: general finite linear primal} satisfies \Cref{assumption: finite support case}.

\begin{proof}
First, since our problem is MOT, the solution set $\gamma(\bm{b}) \subseteq [0,1]^{S_*}$ is always nonempty and bounded. In particular, for generic $\bm{c}$ with vanishing diagonal, i.e., $c(\bm{x},\bm{x})=0$ for all $\bm{x} \in \mathcal{Z}$, $\gamma(\bm{b})$ is $\{ (\mathrm{Id}, \mathrm{Id})_\#(\pi^0) \}$, hence singleton.

A projection matrix $\bm{A}$ has rank $S_+-1$, so it does not have full rank. However, there is a canonical way to adjust \eqref{def: general finite linear primal} so that $\bm{A}$ is of full rank discussed in \cite[Section 6]{limit_random_lp22}. Since $\mu^1, \dots, \mu^K$  are probability measures, it is sufficient to consider $\mu^{1, \dagger}, \dots, \mu^{K, \dagger}$ where $\mu^{\dagger}$ is the truncation of a probability vector obtained by removing the last coordinate. Accordingly, let $\bm{A}^\dagger \in \mathbb{R}^{(S_+ - K) \times S_*}$ denote the matrix obtained from $\bm{A}$ by removing its $s_1$-th, $(s_1 + s_2)$-th, $\dots$ and $(s_1 + \dots, s_K)$-th rows. Then, the linear constraint $\bm{A}\gamma= \bm{b}$ reduces to
\[
    \bm{A}^\dagger \gamma =
\begin{bmatrix}
\bm{\mu}^\dagger\\
\pi^{0}
\end{bmatrix},
\]
by which the full rank assumption is satisfied. To avoid multiple notation, we abuse the notation $\bm{A} \gamma = \bm{b}$ so that it should be understood in the sense as we discussed above.

The limit distribution of $\bm{b}_{n,m}$ is as follows. Assuming $\frac{m}{n+m} \to \lambda \in (0, 1)$ as $n,m \to \infty$, it holds that
\begin{align}\label{eq: limit of mu and pi^0}
    \sqrt{\frac{nm}{n+m}}\left( \bm{b}_{n,m} - \bm{b} \right) \overset{d}{\longrightarrow} \mathbb{G}_{\bm{\mu}, \pi^0} =
\begin{bmatrix}
\sqrt{\lambda} \mathbb{G}_{\bm{\mu}}\\
\sqrt{1-\lambda} \mathbb{G}_{\pi^0}
\end{bmatrix} 
\end{align}
where $\mathbb{G}_{\bm{\mu}}$ is the centered Gaussian on $\mathbb{R}^{s_0}$ with the covariance matrix
\begin{equation*}
    \begin{bmatrix}
\Sigma(\mu_1) \\
 & \Sigma(\mu_2)   \\
 &  & \ddots
&  \\
 & & 
& \Sigma(\mu_K)
\end{bmatrix}
\end{equation*}
with $\Sigma(\mu^i)$, the $s_i \times s_i$ covariance matrix of multinomial distribution $\mu^i$, and $\mathbb{G}_{\pi^0}$ is the centered Gaussian on $\mathbb{R}^{s_0}$ with the covariance matrix $\Sigma(\pi^0)$, which is the $s_0 \times s_0$ covariance matrix of multinomial distribution $\pi^0$. The ratio $\lambda$ is natural: see \cite[Remark 6.2]{limit_random_lp22} for more details.

The last one is the Slater's condition. Such a $\gamma_0$ in \Cref{assumption: finite support case} is sometimes called strictly feasible. Slater’s condition implies that strong duality holds, and the problem becomes convex: See \cite[Section 5.2.3]{boyd2004convex} for more details. This condition is trivially satisfied in this setting by $\gamma_0 = (\otimes \bm{\mu}_n) \otimes \pi^0$. Therefore, \Cref{assumption: finite support case} for \eqref{def: general finite linear primal} is satisfied. 

\end{proof}
\end{lemma}

\subsection{The asymptotic confidence sets}\label{subsec: confidence sets}
In the below, we present the asymptotic confidence sets for $\hat{\pi}$ from \cite{liu2023asymptotic}.

Let $[s]:= \{1, \dots, s\}$. Recall that the columns of $\bm{A}$ can be corresponded by $[s_1] \times \dots \times [s_K] \times [s_0]$ (the length of the column of $\bm{A}$ is $S_*$), which is arranged in lexicographical order. For any subset $\mathrm{I} \subseteq [s_1] \times \dots \times [s_K] \times [s_0]$, we denote by $\bm{A}_\mathrm{I}$ the $S_+ \times |\mathrm{I}|$ submatrix of $\bm{A}$ obtained by taking the columns of $\bm{A}$ which are contained in $\mathrm{I}$.
Analogously, for $\gamma \in \mathbb{R}^{S_*}$, we denote by $\gamma_\mathrm{I}$ the vector of length $|\mathrm{I}|$ consisting of the coordinates of $\gamma$ contained in $\mathrm{I}$.

\begin{definition}
A set $\mathrm{I} \subseteq [s_1] \times \dots \times [s_K] \times [s_0]$ is a basis if
\begin{equation*}
	|\mathrm{I}|=k, \quad \operatorname{rank}(\bm{A}_\mathrm{I})=k.\label{Intro:index_set}
\end{equation*}
Given a basis $\mathrm{I}$, $\gamma(\mathrm{I};\bm{b})$ is called the basic solution which is the vector $\gamma$ satisfying 
\begin{equation*}\label{basis_def}
    \gamma(\mathrm{I};\bm{b})_\mathrm{I} = \bm{A}_\mathrm{I}^{-1} \bm{b}, \quad \gamma(\mathrm{I};\bm{b})_{\mathrm{I}^c} = \bm{0}.
\end{equation*} 
If it is feasible, i.e., $\bm{A}_\mathrm{I}^{-1} \bm{b} \geq \bm{0}$, $\gamma(\mathrm{I};\bm{b})$ is called a basic feasible solution.
\end{definition}

\begin{definition}
\begin{align*}
    \mathcal{I}(\bm{b}) &:= \left\{ \mathrm{I} :  \text{$\mathrm{I}$ is a basis, }  \text{$\gamma(\mathrm{I}; \bm{b})$ is a basic feasible solution} \right\},\\
    \mathcal{I}^*(\bm{b}) &:= \left\{ \mathrm{I} \in \mathcal{I}(\bm{b}) : \text{optimal for \eqref{def: general finite linear primal}}\right\}.
\end{align*}
The set of optimal basic solutions of \eqref{def: general finite linear primal} is denoted by
\begin{equation*}
	\mathcal{V}(\bm{b}) := \{\gamma(\mathrm{I}; \bm{b}): \mathrm{I} \in \mathcal{I}^*(\bm{b})\}.
\end{equation*}
\end{definition}

Let $\gamma(\bm{b})$ be a set of optimal solutions for \eqref{def: general finite linear primal}. In terms of linear programming language, $\mathcal{V}(\bm{b})$ is the set of extreme solutions, or vertex solutions, of which elements cannot be written as any convex combination of elements of $\gamma(\bm{b})$. In fact, $\gamma(\bm{b}) =  \operatorname{conv}(\mathcal{V}(\bm{b}))$, the convex hull of $\mathcal{V}(\bm{b})$. Notice that $\gamma(\bm{b})$ is always bounded if $\bm{b}$ is the marginal (probability) constraint vector since \eqref{def: general finite linear primal} is (multimarginal) optimal transport problem.

Now, suppose a statistician solved \eqref{def: general finite linear primal} with an empirical input $\bm{b}_{n,m}$ and obtained a (random) solution $\hat \gamma(\bm{b}_{n,m}) \in \mathcal{V}(\bm{b}_{n,m})$ such that a corresponding basis $\mathrm{I}_n \in \mathcal{I}^*(\bm{b}_{n,m})$. The goal is to construct a confidence set based on $\hat \gamma(\bm{b}_{n,m})$ that will contain at least one element of $\gamma(\bm{b})$ with high probability. 
Let $\gamma(\mathrm{I}_n; \bm{b})$ be the basic solution  obtained by random basis $\mathrm{I}_n$, which may be neither feasible for \eqref{def: general finite linear primal} nor optimal. To remedy these issues, define the projection
\begin{equation}\label{proj}
	\overline{\gamma}^*_n :=  \argmin_{ \gamma \in \gamma(\bm{b})} \| \gamma(\mathrm{I}_n; \bm{b}) -\gamma \|\,.
\end{equation}

The following result is about the confidence set for \eqref{def: general finite linear primal}.

\begin{lemma}{\cite[Theorem 8]{liu2023asymptotic}}\label{lem: confidence_set}
Suppose that \eqref{def: general finite linear primal} satisfies \cref{assumption: finite support case} and $\bm{b}_{n,m}$ satisfies the distributional limit $\mathbb{G}$. Let $G_\alpha$ be an open set such that $\mathbb{P} \{\mathbb{G} \in G_\alpha \} \geq 1 - \alpha$.
	Then
	\begin{equation}
		\liminf_{n \to \infty} \mathbb{P}\left\{ r_{n,m}(\hat\gamma(\bm{b}_{n,m}) - \overline{\gamma}^*_n)\in \gamma(\mathrm{I}_n;G_\alpha) \right\} \geq 1-\alpha\,,
		\label{Confidence set}
	\end{equation}
where $\gamma(\mathrm{I}_n;G_\alpha) := \{\gamma(\mathrm{I}_n;\bm{G}): \bm{G} \in G_\alpha\}$.
\end{lemma}

\begin{corollary}[Confidence set for an optimal solution] {\cite[Corollary 9]{liu2023asymptotic}}\label{cor: confidence_cor}
In the setting of \Cref{lem: confidence_set}, the set $C_n := \{\hat \gamma(\bm{b}_{n,m}) - r_{n,m}^{-1} \gamma: \gamma \in \gamma(\mathrm{I}_n;G_\alpha)\}$ contains an element of $\gamma(\bm{b})$ with asymptotic probability at least $1-\alpha$.
\end{corollary}

 From \Cref{cor: confidence_cor}, one can obtain asymptotic confidence sets containing $\pi^0$ with probability at least $1-\alpha$ by projection. We omit the proof.

\begin{theorem}[Confidence set]\label{thm: confidence sets pi^0}
Fix $\alpha \in (0,1)$. Assume that $\frac{m}{n+m} \to \lambda \in (0, 1)$ as $n,m \to \infty$ and $c=d_{\mathcal{Z},p}^p$ in \eqref{eq: metric over Z}. Let $\hat{\pi}(\bm{b}_{n,m})$ be a solution for \eqref{eq: Wasserstein projection} with input $\bm{\mu}_n$ and $\pi^0_m$. Let $G_\alpha$ be an open set such that $\mathbb{P} \{\mathbb{G}_{\bm{\mu}, \pi^0} \in G_\alpha \} \geq 1 - \alpha$ where $\mathbb{G}_{\bm{\mu}, \pi^0}$ is the limit distribution of $\bm{b}_{n,m}$ given in \eqref{eq: limit of b_n}. A confidence set
\[
    C_\alpha:= \left\{ \hat{\pi}(\bm{b}_{n,m}) - \sqrt{\frac{n+m}{nm}} \pi: \pi \in \pi(\tilde{\mathrm{I}}; G_\alpha) \right\}
\]
contains $\pi^0$ with asymptotic probability at least $1-\alpha$ where $\tilde{\mathrm{I}} \in \mathcal{I}^*(\bm{b}_{n,m})$ and $\pi(\tilde{\mathrm{I}}; G_\alpha) := \{ (\mathrm{P}_{\bm{x}})_\# \gamma(\tilde{\mathrm{I}};\bm{G}): \bm{G} \in G_\alpha \}$.
\end{theorem}

\section{Appendix: Additional Simulation Results} \label{sec:appendix:additional sim}

Table \ref{tab:Sim2} presents the root mean squared error (RMSE) of $\widetilde{F}(-0.25,-0.25)$ and $\widehat{F}(-0.25,-0.25)$, averaged over 500 repetitions, for each combination of $(m, n, \rho)$. The results are largely consistent with those in Table \ref{tab:Sim1}. First, for each $\rho$ and fixed ratio $n/m$, the RMSE values of $\widetilde{F}$ and $\widehat{F}$ decreases with $m$, approximately at the rate of $1/\sqrt{m}$. Second, the optimal $\eta$ depends on $\rho$, with smaller $\eta$ preferred for large $\rho$ and larger $\eta$ for small $\rho$. Third, the cross-validation procedure yields $\widehat{F}$ that outperforms $\widetilde{F}$, demonstrating the practical advantage of incorporating marginal data for more efficient estimation.

\begin{table}[!htp]
\scriptsize
\renewcommand{\arraystretch}{1.2}
\centering
\setlength{\tabcolsep}{6pt}
\begin{tabular}{|c|c|c|cccc|c|}
\hline
\multirow{3}{*}{$\rho$} & \multirow{3}{*}{$n/m$} & \multirow{3}{*}{$m$} & \multicolumn{4}{c|}{RMSE}                                                                                                                              & \multirow{3}{*}{Average $\log_2 \eta_{\text{cv}}$} \\ \cline{4-7}
                        &                        &                      & \multicolumn{1}{c|}{\multirow{2}{*}{$\widetilde{F}$}} & \multicolumn{3}{c|}{$\widehat{F}$}                                                             &                                                    \\ \cline{5-7}
                        &                        &                      & \multicolumn{1}{c|}{}                                 & \multicolumn{1}{c|}{$\eta=2^0$} & \multicolumn{1}{c|}{$\eta=2^{-6}$} & $\eta=\eta_{\text{cv}}$ &                                                    \\ \hline
\multirow{9}{*}{0.00}   & \multirow{3}{*}{5}     & 50                   & \multicolumn{1}{c|}{5.18}                             & \multicolumn{1}{c|}{2.29}       & \multicolumn{1}{c|}{3.22}          & 3.20                    & -2.63                                              \\ \cline{3-8} 
                        &                        & 100                  & \multicolumn{1}{c|}{3.65}                             & \multicolumn{1}{c|}{1.63}       & \multicolumn{1}{c|}{2.36}          & 2.29                    & -2.75                                              \\ \cline{3-8} 
                        &                        & 200                  & \multicolumn{1}{c|}{2.55}                             & \multicolumn{1}{c|}{1.06}       & \multicolumn{1}{c|}{1.48}          & 1.42                    & -2.71                                              \\ \cline{2-8} 
                        & \multirow{3}{*}{10}    & 50                   & \multicolumn{1}{c|}{5.18}                             & \multicolumn{1}{c|}{1.86}       & \multicolumn{1}{c|}{2.94}          & 2.87                    & -2.53                                              \\ \cline{3-8} 
                        &                        & 100                  & \multicolumn{1}{c|}{3.65}                             & \multicolumn{1}{c|}{1.34}       & \multicolumn{1}{c|}{2.14}          & 2.16                    & -3.08                                              \\ \cline{3-8} 
                        &                        & 200                  & \multicolumn{1}{c|}{2.55}                             & \multicolumn{1}{c|}{0.92}       & \multicolumn{1}{c|}{1.39}          & 1.35                    & -2.70                                              \\ \cline{2-8} 
                        & \multirow{3}{*}{25}    & 50                   & \multicolumn{1}{c|}{5.18}                             & \multicolumn{1}{c|}{1.55}       & \multicolumn{1}{c|}{2.72}          & 2.77                    & -2.63                                              \\ \cline{3-8} 
                        &                        & 100                  & \multicolumn{1}{c|}{3.65}                             & \multicolumn{1}{c|}{1.17}       & \multicolumn{1}{c|}{2.04}          & 1.99                    & -2.83                                              \\ \cline{3-8} 
                        &                        & 200                  & \multicolumn{1}{c|}{2.55}                             & \multicolumn{1}{c|}{0.80}       & \multicolumn{1}{c|}{1.34}          & 1.29                    & -2.70                                              \\ \hline
\multirow{9}{*}{0.25}   & \multirow{3}{*}{5}     & 50                   & \multicolumn{1}{c|}{5.46}                             & \multicolumn{1}{c|}{2.83}       & \multicolumn{1}{c|}{3.32}          & 3.53                    & -4.03                                              \\ \cline{3-8} 
                        &                        & 100                  & \multicolumn{1}{c|}{4.01}                             & \multicolumn{1}{c|}{2.30}       & \multicolumn{1}{c|}{2.47}          & 2.71                    & -4.49                                              \\ \cline{3-8} 
                        &                        & 200                  & \multicolumn{1}{c|}{2.84}                             & \multicolumn{1}{c|}{1.91}       & \multicolumn{1}{c|}{1.60}          & 1.76                    & -5.21                                              \\ \cline{2-8} 
                        & \multirow{3}{*}{10}    & 50                   & \multicolumn{1}{c|}{5.46}                             & \multicolumn{1}{c|}{2.41}       & \multicolumn{1}{c|}{2.96}          & 3.27                    & -3.79                                              \\ \cline{3-8} 
                        &                        & 100                  & \multicolumn{1}{c|}{4.01}                             & \multicolumn{1}{c|}{2.11}       & \multicolumn{1}{c|}{2.24}          & 2.48                    & -4.51                                              \\ \cline{3-8} 
                        &                        & 200                  & \multicolumn{1}{c|}{2.84}                             & \multicolumn{1}{c|}{1.77}       & \multicolumn{1}{c|}{1.47}          & 1.66                    & -5.22                                              \\ \cline{2-8} 
                        & \multirow{3}{*}{25}    & 50                   & \multicolumn{1}{c|}{5.46}                             & \multicolumn{1}{c|}{2.19}       & \multicolumn{1}{c|}{2.69}          & 3.04                    & -3.82                                              \\ \cline{3-8} 
                        &                        & 100                  & \multicolumn{1}{c|}{4.01}                             & \multicolumn{1}{c|}{1.98}       & \multicolumn{1}{c|}{2.13}          & 2.41                    & -4.67                                              \\ \cline{3-8} 
                        &                        & 200                  & \multicolumn{1}{c|}{2.84}                             & \multicolumn{1}{c|}{1.69}       & \multicolumn{1}{c|}{1.39}          & 1.57                    & -5.14                                              \\ \hline
\multirow{9}{*}{0.75}   & \multirow{3}{*}{5}     & 50                   & \multicolumn{1}{c|}{6.26}                             & \multicolumn{1}{c|}{6.34}       & \multicolumn{1}{c|}{3.42}          & 3.60                    & -5.77                                              \\ \cline{3-8} 
                        &                        & 100                  & \multicolumn{1}{c|}{4.33}                             & \multicolumn{1}{c|}{6.09}       & \multicolumn{1}{c|}{2.60}          & 2.63                    & -5.97                                              \\ \cline{3-8} 
                        &                        & 200                  & \multicolumn{1}{c|}{3.10}                             & \multicolumn{1}{c|}{5.88}       & \multicolumn{1}{c|}{1.92}          & 1.80                    & -6.00                                              \\ \cline{2-8} 
                        & \multirow{3}{*}{10}    & 50                   & \multicolumn{1}{c|}{6.26}                             & \multicolumn{1}{c|}{6.08}       & \multicolumn{1}{c|}{2.95}          & 3.18                    & -5.73                                              \\ \cline{3-8} 
                        &                        & 100                  & \multicolumn{1}{c|}{4.33}                             & \multicolumn{1}{c|}{5.95}       & \multicolumn{1}{c|}{2.21}          & 2.22                    & -5.99                                              \\ \cline{3-8} 
                        &                        & 200                  & \multicolumn{1}{c|}{3.10}                             & \multicolumn{1}{c|}{5.82}       & \multicolumn{1}{c|}{1.67}          & 1.53                    & -6.00                                              \\ \cline{2-8} 
                        & \multirow{3}{*}{25}    & 50                   & \multicolumn{1}{c|}{6.26}                             & \multicolumn{1}{c|}{5.99}       & \multicolumn{1}{c|}{2.57}          & 2.80                    & -5.79                                              \\ \cline{3-8} 
                        &                        & 100                  & \multicolumn{1}{c|}{4.33}                             & \multicolumn{1}{c|}{5.93}       & \multicolumn{1}{c|}{2.01}          & 1.99                    & -5.99                                              \\ \cline{3-8} 
                        &                        & 200                  & \multicolumn{1}{c|}{3.10}                             & \multicolumn{1}{c|}{5.77}       & \multicolumn{1}{c|}{1.54}          & 1.39                    & -6.00                                              \\ \hline
\end{tabular}
\caption{\scriptsize Each cell reports the RMSE of the estimators for $F(-0.25,-0.25)$, computed over 500 simulation repetitions. The RMSE values are scaled by a factor of 100.  The last column, labeled ``Average $\log_{2} \eta_{\text{cv}}$,'' reports the mean of $\log_{2}\eta_{\text{cv}}$ across 500 repetitions.}
\label{tab:Sim2}
\end{table}

\newpage

\end{document}